\documentclass[lettersize,journal]{IEEEtran}
\usepackage{amsfonts}
\usepackage{amssymb}
\usepackage{amsthm}
\usepackage{amsmath,amsfonts,amssymb}
\usepackage{mathrsfs}
\ifCLASSINFOpdf
\usepackage[pdftex]{graphicx}
\else
\fi
\ifCLASSOPTIONcompsoc
\usepackage[caption=false,font=normalsize,labelfont=sf,textfont=sf]{subfig}
\else
\usepackage[caption=false,font=footnotesize]{subfig}
\fi

\usepackage{verbatim}
\usepackage{setspace}
\usepackage{bm}
\usepackage{bbm}
\usepackage{algorithm}
\usepackage{algorithmic}
\ifCLASSOPTIONcompsoc
\else
\usepackage[caption=false,font=footnotesize]{subfig}
\fi
\usepackage{cite}
\usepackage{color}
\usepackage{multirow}
\usepackage{setspace}

\usepackage{changepage}
\usepackage{pdfpages}
\usepackage{color}

\newtheorem{lemma}{Lemma}

\usepackage{flushend} %最后一页，底部对齐

\usepackage{gensymb}
\usepackage{booktabs}
\usepackage{threeparttable}
\usepackage{tabularx}
\usepackage{multirow}
\usepackage{array}
\newcommand{\PreserveBackslash}[1]{\let\temp=\\#1\let\\=\temp}
\newcolumntype{C}[1]{>{\PreserveBackslash\centering}p{#1}}
\newcolumntype{R}[1]{>{\PreserveBackslash\raggedleft}p{#1}}
\newcolumntype{L}[1]{>{\PreserveBackslash\raggedright}p{#1}}

\usepackage{url}
\makeatletter
\renewcommand{\maketag@@@}[1]{\hbox{\m@th\normalsize\normalfont#1}}%
\makeatother

\IEEEoverridecommandlockouts
\begin{document}
% 激活bst控制：主要包括超过6个以上的作者使用et al.等控制
\bstctlcite{IEEEexample:BSTcontrol} 
\title{Movable-Antenna Aided Secure Transmission \\for RIS-ISAC Systems}
\author{Yaodong~Ma,~\IEEEmembership{Graduate Student Member,~IEEE,}
	Kai~Liu,~\IEEEmembership{Member,~IEEE,}
	Yanming~Liu,~\IEEEmembership{Member,~IEEE,}
	Lipeng~Zhu,~\IEEEmembership{Member,~IEEE,} and 
	Zhenyu~Xiao,~\IEEEmembership{Senior Member,~IEEE}
	% <-this % stops a space
	\thanks{This work was supported by the National Nature Science Foundation of China under Grant No. U2233216.  \textit{(Corresponding author: Lipeng Zhu)}
		
		Y. Ma, K. Liu, Y. Liu, and Z. Xiao are with the School of Electronics and Information Engineering, Beihang University, Beijing, 100191, China, and the State Key Laboratory of CNS/ATM, Beijing, 100191, China (e-mail: \{yaodongma, liuk, liuyanming, xiaozy\}@buaa.edu.cn).
		
		L. Zhu is with the Department of Electrical and Computer Engineering, National
		University of Singapore, Singapore 117583 (e-mail: zhulp@nus.edu.sg).
	}
	% <-this % stops a space
	\thanks{}}
\vspace{-10mm}
\maketitle
\vspace{-10mm}

\begin{abstract}
Integrated sensing and communication (ISAC) systems have the issue of secrecy leakage when using the ISAC waveforms for sensing, thus posing a potential risk for eavesdropping. To address this problem, we propose to employ movable antennas (MAs) and reconfigurable intelligent surface (RIS) to enhance the physical layer security (PLS) performance of ISAC systems, where an eavesdropping target potentially wiretaps the signals transmitted by the base station (BS). To evaluate the synergistic performance gain provided by MAs and RIS, we formulate an optimization problem for maximizing the sum-rate of the users by jointly optimizing the transmit/receive beamformers of the BS, the reflection coefficients of the RIS, and the positions of MAs at communication users, subject to a minimum communication rate requirement for each user, a minimum radar sensing requirement, and a maximum secrecy leakage to the eavesdropping target. To solve this non-convex problem with highly coupled variables, a two-layer penalty-based algorithm is developed by updating the penalty parameter in the outer-layer iterations to achieve a trade-off between the optimality and feasibility of the solution. In the inner-layer iterations, the auxiliary variables are first obtained with semi-closed-form solutions using Lagrange duality. Then, the receive beamformer filter at the BS is optimized by solving a Rayleigh-quotient subproblem. Subsequently, the transmit beamformer matrix is obtained by solving a convex subproblem. Finally, the majorization-minimization (MM) algorithm is employed to optimize the RIS reflection coefficients and the positions of MAs. Extensive simulation results validate the considerable benefits of the proposed MAs-aided RIS-ISAC systems in enhancing security performance compared to traditional fixed position antenna (FPA)-based systems. 
%In addition, it is revealed that the increase in the size of receive region, the number of paths, and the number of RIS elements can enhance the security of the MA-aided system. (need?)
\end{abstract} %226words

\begin{IEEEkeywords}
Integrated sensing and communication (ISAC), movable antenna (MA), reconfigurable intelligent surface (RIS), physical-layer security (PLS).
\end{IEEEkeywords}

% \fi
\IEEEpeerreviewmaketitle

\section{Introduction}
\IEEEPARstart{B}{oth} high-quality communications and high-accuracy sensing capabilities are essential in the six-generation (6G) wireless networks to facilitate  auto-driving, extended reality, and other emerging applications \cite{liu2022integrated}. Therefore, an important paradigm shift is needed from communication-oriented systems to the networks with integrated sensing and communication (ISAC) \cite{zhang2021enabling,zheng2019radar}. 
%sensing-communication integration, which provokes the emergence and development of integrated sensing and communication (ISAC) \cite{zhang2021enabling}. 
%Due to similar hardware platforms, signal processing algorithms, and the same evolution direction towards high-frequency wideband multi-antenna systems, the ISAC systems have evolved from coexistence, cooperation, to co-design \cite{zheng2019radar}. 
By enabling spectrum resource sharing and employing a unified platform for dual-functional waveform transmission to simultaneously perform sensing and communication, ISAC aims at greatly enhancing the spectral and energy efficiency as well as reducing hardware costs and signaling overhead. Nevertheless, several critical challenges still need to be addressed. For example, due to the complicated radio propagation environments, the practical ISAC performance may degrade seriously when the communication/sensing links are blocked by obstacles such as vehicles and buildings \cite{liu2022survey}.

%---------------------- RIS-ISAC ---------------------
Recently, reconfigurable intelligent surface (RIS), also known as intelligent reflecting surface (IRS), has been recognized as a promising technology for 6G \cite{wu2023intelligent}, which can not only reconfigure the wireless channels but also provide the line-of-sight (LoS) links for both communication and sensing. Thus, it introduces additional degrees of freedom (DoFs) to design future ISAC systems, which has been widely investigated to improve the system capacity \cite{liu2022proximal}, mitigate echo interference \cite{cao2023feedback}, and enhance channel estimation performance \cite{chen2024joint}.
Besides, RIS has been employed to detect non-LoS (NLoS) multiple targets and extend the coverage of communication devices in ISAC \cite{song2022joint}.
Moreover, the authors in \cite{liu2022joint,zhu2023joint} jointly optimized RIS coefficients and transmit/receive beamformers for maximizing the sensing performance, subject to the communication requirements of the users. Similarly, the authors in \cite{hao2024joint} maximized the sum-rate of users while satisfying a minimum radar signal-to-noise ratio (SNR) constraint.
% by jointly optimizing the RIS phase shifts and the transmit beamformers. 
%In general, the main idea of most existing RIS-ISAC designs was to formulate and solve specific problems aiming at improving one performance metric (either communication or sensing) while guaranteeing the minimum performance requirements of the other, thus achieving a trade-off between communication and sensing.
Nevertheless, the above studies \cite{wu2023intelligent,liu2022proximal,cao2023feedback,chen2024joint,song2022joint,liu2022joint,zhu2023joint,hao2024joint} considered that the targets would not wiretap the transmitted information. In practice, there exists a potential security issue for ISAC, where transmitted signals may expose private information to untrusted detection targets (e.g., unauthorized unmanned aerial vehicles (UAVs)) \cite{wei2022toward}. Thus, the trade-off between maintaining the ISAC performance and reducing the risk of secrecy leakage to the eavesdropping targets should be carefully balanced.

%--------------------- secure ISAC&RIS-ISAC ----------------
To guarantee the security, many existing works have studied the physical-layer security (PLS) in ISAC systems.
Assuming the target was an eavesdropper, the authors in \cite{su2022secure} employed the interference to confuse the eavesdropper, thereby improving security performance of ISAC. In \cite{luo2022secure}, the authors investigated both the sum-rate and jamming power maximization problems. Subsequently, the imperfect channel state information (CSI) \cite{liu2022outage} and possible estimation errors \cite{su2023sensing} were further considered in the transmit beamforming design to guarantee transmission secrecy.
In addition, a few recent studies \cite{hua2023secure,salem2023active,li2024noma,zhang2024secure} have focused on secure transmission designs for RIS-ISAC systems. Specifically, the authors in \cite{hua2023secure} developed the Lagrange duality and majorization-minimization (MM) algorithms to maximize the sensing beampattern gain while guaranteeing communication and PLS requirements. In \cite{salem2023active}, the authors considered security issues for a RIS-ISAC system when being eavesdropped by a malicious UAV, and the achievable secrecy rate was maximized by jointly optimizing the radar receiving beamformer, active RIS reflection coefficients, and transmit beamformer. In \cite{li2024noma}, the authors  jointly optimized the jamming, active transmit precoding, and passive phase reflecting coefficients for maximizing the sum secrecy rate. The authors in \cite{zhang2024secure} proposed an alternating optimization (AO) algorithm for securing transmission against a potential eavesdropper. 
However, the above downlink RIS-ISAC studies mainly focused on the transceiver/reflection design at the BS/RIS by employing fixed-position antennas (FPAs), while the channel spatial variation was not exploited therein, leading to performance limitations.

%only reconfigured the wireless channels at the transmitter (Tx) to improve the transmission security, while the channel spatial variation at the receiver (Rx) was not exploited therein, leading to performance limitations.

%--------------------------- MA -------------------------
To overcome the performance limitations of conventional FPAs, the concept of the movable antenna (MA) 
%also known as fluid antenna system, 
has been introduced to wireless systems \cite{zhu2024historical}, allowing the positions of antenna elements at the transceiver to be flexibly adjusted to enhance communication performance \cite{zhu2023modeling, ma2023mimo}.
% (also known as fluid antenna system (FAS) \cite{zhu2024historical}).
Specifically, flexible cables can connect MAs to the radio frequency (RF) chains \cite{zhu2023movable}, enabling the reconstruction of wireless channels by leveraging new DoFs in the spatial domain.
The field-response-based channel model for MAs-assisted wireless communication systems was developed in \cite{zhu2023modeling}, and the authors in \cite{ma2023mimo} studied the capacity maximization problem for MAs-enabled multiple-input multiple-output (MIMO) systems. Both works have demonstrated the significant performance gain provided by antenna movement in terms of spatial diversity and multiplexing. Based on this channel model, many existing works have highlighted the benefits of MAs-aided systems compared to FPA systems in improving signal-to-interference-plus-noise ratio (SINR) of received signals \cite{gao2024joint,xiao2023multiuser,mei2024movable,wei2024joint}, mitigating interference \cite{zhu2023movableM}, reducing transmit powers \cite{zhu2023movable,wu2023movable}, and enhancing wireless sensing performance \cite{ma2024movable}. 
In addition, several recent works have integrated MAs into ISAC systems to enhance the beampattern gain \cite{wu2024movable} or the communication rate under the constraint of sensing requirements \cite{kuang2024movable,lyu2024flexible}, while security performance was not considered therein.
In terms of PLS,
%enhancing system security performance \cite{cheng2024enabling,tang2024secure,hu2024movable,hu2024secure}. 
the authors in \cite{cheng2024enabling, tang2024secure} jointly optimized the transmit beamformer and the positions of MAs to maximize the system secrecy rate in the presence of an eavesdropper. 
%Besides, the authors in \cite{hu2024movable} proposed an alternating projected gradient ascent algorithm to minimize the secrecy outage probability by jointly optimizing the transmit beamforming and positions of MAs with multiple eavesdroppers. In \cite{hu2024secure}, the authors developed an AO algorithm to jointly optimize the transmit beamforming and the positions of MAs for maximizing the achievable secrecy rate in the presence of multiple single-antenna eavesdroppers.
Besides, considering multiple single-antenna eavesdroppers, the authors in \cite{hu2024movable,hu2024secure} proposed the alternating projected gradient ascent algorithm to 
jointly optimize the positions of MAs and the transmit beamformer for enhancing system PLS.
Although works \cite{cheng2024enabling,tang2024secure,hu2024secure,hu2024movable} investigated secure transmission designs for various communication systems, the role of MAs in securing transmission for RIS-ISAC systems was not unveiled, and the transceiver design in the literature on MA/RIS-aided communication systems cannot be directly applied.
%Meanwhile,  research on securing RIS-ISAC is still in its infancy.%Moreover, the authors in [x] made an important initial attempt to characterize the fundamental performance of MA-aided wireless sensing.

%---------------- Motivation & Contribution -----------------
Motivated by the above discussions, to fully exploit the DoFs in channel reconfiguration provided by the MAs and RIS, this paper investigates MAs-assisted secure transmission for RIS-ISAC systems, where the base station (BS) transmits private information to multiple single MA-enabled users under the threat of an eavesdropping target.
%In addition, the RIS is leveraged to not only create a virtual LoS link for target sensing but also to assist downlink communications from the BS to multiple users.
%To the best of our knowledge, this is the first attempt to explore secure transmission issues for MAs-assisted RIS-ISAC systems. 
The main contributions of this paper are summarized as follows:
\begin{itemize}
	\item We consider an MAs-aided RIS-ISAC system in the presence of an eavesdropping target, aiming at enhancing the PLS. Then, we formulate an optimization problem to maximize the sum-rate of all users by jointly optimizing the transmit/receive beamformers, the RIS phase shifts, and the positions of MAs at communication users, while ensuring the minimum communication rate requirements of users, the minimum sensing SNR constraint of the target, and the maximum secrecy leakage constraint.
%	 to the eavesdropping target.
	\item A two-layer iterative algorithm based on the penalty method is proposed to tackle the formulated non-convex optimization problem. In the inner-layer iterations, for a given penalty factor, we optimize the variables within different blocks in an alternating manner. Specifically, by using Lagrange duality, we first obtain the auxiliary variables with a semi-closed-form solution. Then, the receive beamformer filter at the BS is optimally obtained in a closed form by solving a Rayleigh quotient subproblem. Subsequently, the transmit beamformer matrix is obtained by solving a convex subproblem. Finally, the RIS reflection coefficients and the positions of MAs are obtained by utilizing the MM algorithm. In the outer-layer iterations, the penalty parameter is updated for achieving a trade-off between the optimality and feasibility of the solution.
	\item Extensive simulation results are shown to validate the considerable benefits of the MAs-assisted schemes in terms of achieving higher sum-rates compared to other conventional FPA systems. This also indicates that the MA-aided schemes can achieve enhanced security performance under a given secrecy constraint.
\end{itemize}

The rest of the paper is organized as follows. In Section II, the system model and problem formulation are presented. In Section III, we propose the two-layer penalty-based solution.  The simulations are presented in Section VI to demonstrate the performance of the proposed solution and the conclusions are shown in Section VII.

\textit{Notation}: The bold-face lower-case and upper-case letters denote the vectors and matrices, respectively. $|a|$ denotes the magnitude of scalar $a$. The norm of vector $\mathbf{a}$ is defined by $\|\mathbf{a}\|$. $\|\mathbf{A}\|_F$ denotes the Frobenius norm of matrix $\mathbf{A}$. $(\cdot)^{\rm T}$, $(\cdot)^{*}$, and $(\cdot)^{\rm H}$ represent the transpose, conjugate operations, and Hermitian transpose, respectively. $\mathbf{I}_N$ indicates the $N$-order identity matrix. In addition, we denote a complex circularly-symmetric Gaussian distribution with mean zero and covariance $b$ as $\mathcal{CN}\left(0,b\right)$. ${\rm diag}\left\{\mathbf{a}\right\}$ represents a diagonal matrix with the elements of vector $\mathbf{a}$ on the main diagonal. $\otimes$ denotes the Kronecker product operator.  The gradient vector of function $f$ with respect to (w.r.t.) vector $\mathbf{x}$ is denoted by $\nabla_{\mathbf{x}}f(\mathbf{x})$, and $\nabla^2_{\mathbf{x}}f(\mathbf{x})$ denotes the Hessian matrix of function $f$ w.r.t. vector $\mathbf{x}$. The real and imaginary components of a complex number are represented by $\Re\{\cdot\}$ and $\Im\{\cdot\}$, respectively. The phase of complex number $a$ is denoted by $\angle{a}$. The symbol $\mathbb{E}\{\cdot\}$ represents the statistical expectation.

\section{System Model and Problem Formulation}
\begin{figure}[t]
	\begin{center}
		\includegraphics[width= 3 in]{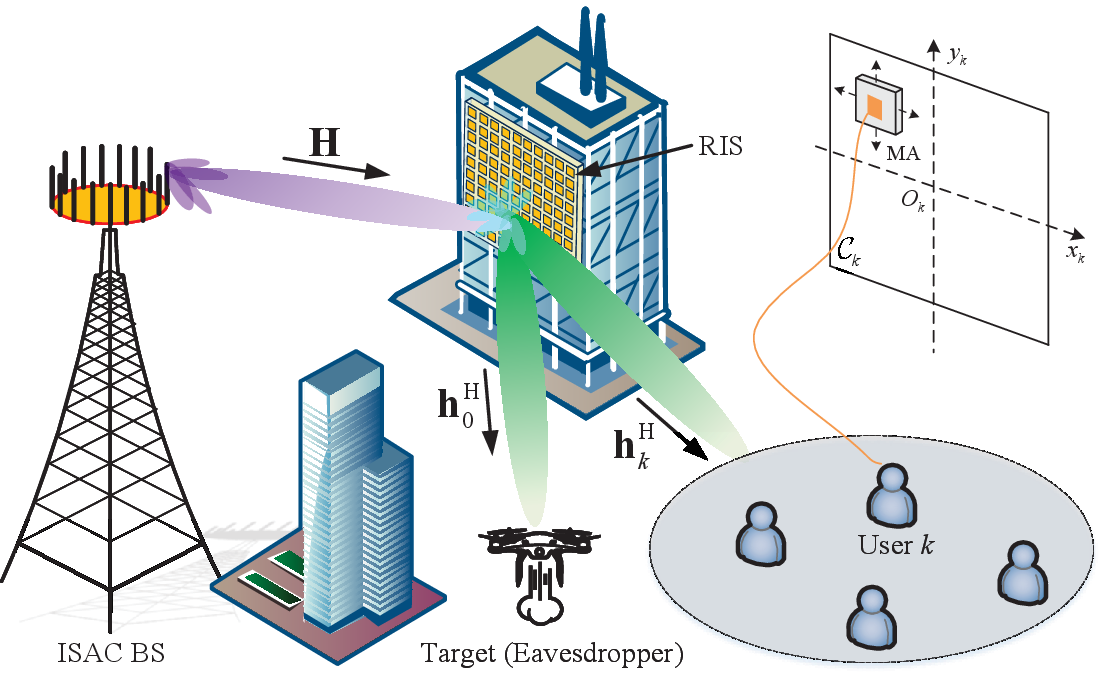}
		\caption{Illustration of the RIS-ISAC system with $K$ single-MA users.} \label{fig::systemModel}
	\end{center}
	\vspace{-1.5em}
\end{figure}

In this paper, we consider a RIS-ISAC system as illustrated in Fig. \ref{fig::systemModel}, where a dual-function BS transmits information to $K$ single-MA users, while an eavesdropping target is hovering nearby to intercept their data\footnote{Although we consider a single eavesdropping target in this paper, both the problem formulation and proposed algorithm can be readily extended to multiple eavesdropping targets with slight modifications.}. 
%The ISAC BS enables its radar sensing functionality for detecting the eavesdropping target and hence, the transmit/receive beamformers, RIS capabilities and MA positions are jointly utilized to hinder the eavesdropping operation. 
Specifically, the uniform planar array (UPA) BS is equipped with $M_t$ transmit FPAs and $M_r$ receive FPAs, where we assume $M_t = M_r = M$ for simplicity.
%a uniform planar array (UPA) of size $M = M_1 \times M_2$ transmit antennas, where $M_1$ and $M_2$ denote the number of antennas along horizontal and vertical directions, respectively. 
It is assumed that the direct links from the BS to the users/target are not available due to blockages, thus the RIS is responsible for creating strong virtual LoS links to 
ensure communication and sensing performance.
The RIS is formed by stacking $N = N_1 \times N_2$ ($N_1$ horizontal elements and $N_2$ vertical elements) passive reflecting elements. To improve the channel conditions for user $k$, the MA can be moved within a local rectangular area, denoted as $\mathcal{C}_k = [x_k^{\min},y_k^{\min}] \times [x_k^{\max},y_k^{\max}]$. Moreover, for each user $k$, a local coordinate system, i.e., $\mathbf{u}_k = [x_k ,y_k]^{\mathrm{T}} \in \mathcal{C}_k$ with $1 \leq k \leq K$, is established to specify the MA position. It is noted that the local coordinate systems for various users are independently  defined, with the origin of the $k$-th user being $O_k$. Additionally, the position of the FPA at the eavesdropper is represented by $\mathbf{u}_0 = [x_0, y_0]^{\rm T}$, the local coordinate of the $m$-th FPA at the ISAC BS is represented by $\mathbf{v}_m = [X_m, Y_m]^{\rm T}, 1 \leq m \leq M$, and the position of the $n$-th RIS reflection element is represented by $\mathbf{t}_n = [\tilde{X}_n, \tilde{Y}_n]^{\rm T}, 1 \leq n \leq N$.
%\footnote{According to \cite{bjornson2022reconfigurable}, an RIS can be seen as a synchronized multi-antenna BS with a phased array, where each reconfigurable element of the RIS can be considered analogous to an antenna. Therefore, we define the position of each element to be similar to the FPA at the BS.}.

\subsection{Transmit Signal Model} % delete or not?
%The transmitted signal needs to be tailored against the eavesdropping while fulfilling the communication and sensing demands. We assume that the ISAC BS simultaneously transmits information signals and radar signals in time slot $t$, which is given by
In the considered systems, both communication and radar signals are assumed to be transmitted simultaneously at the BS. As such, the transmit signal vector is given by
\begin{equation} \label{eq::transmitted_signal} \small
		\mathbf{x} = \mathbf{W}_{c}\mathbf{s}_c + \mathbf{W}_{r}\mathbf{s}_r = \mathbf{W}\mathbf{s},
\end{equation}	
where $\mathbf{s}_c = \left[s_{c,1}, \cdots, s_{c,K}\right]^{{\rm T}}$  represents the communication symbols intended for the $K$ users, 
% with $\mathbb{E}\left\{\mathbf{s}_c\mathbf{s}_c^{\rm H}\right\} = \mathbf{I}_K$
and $\mathbf{s}_r = \left[s_{r,1}, \cdots, s_{r,M}\right]^{{\rm T}}$ denotes $M$ individual radar waveforms.
%with $\mathbb{E}\left\{\mathbf{s}_r\mathbf{s}_r^{\rm H}\right\} = \mathbf{I}_M$. 
%It is assumed that communication and radar signals are statistically independent and uncorrelated, i.e., $\mathbb{E}\{\mathbf{s}_c\mathbf{s}^{\rm H}_r\} = \mathbf{0}$. 
$\mathbf{W}_{c} = [\mathbf{w}_{c,1}, \cdots, \mathbf{w}_{c,k}, \cdots, \mathbf{w}_{c,K}] \in \mathbb{C}^{M \times K}$ and $\mathbf{W}_{r} = [\mathbf{w}_{r,1}, \cdots, \mathbf{w}_{r,m}, \cdots, \mathbf{w}_{r,M}] \in \mathbb{C}^{M \times M}$ represent the transmit beamformers for communication and sensing, respectively. Besides, the combined beamforming matrix and signals are defined as $\mathbf{s} \triangleq [\mathbf{s}^{\rm T}_c,\mathbf{s}^{\rm T}_r]^{\rm T} \in \mathbb{C}^{(K + M) \times 1}$ and $\mathbf{W} \triangleq [\mathbf{W}_{c}, \mathbf{W}_{r}]\in \mathbb{C}^{M \times (K + M)}$. 

\subsection{Channel Model}
% 准静态信道+远场假设 (暂时删除)
%We consider narrow-band channels with slow fading and focus on one quasi-static fading block, during which all the channels involved are assumed to remain unchanged. 
Since the signal propagation distance is much larger than the size of the transmit/receive region, the far-field condition is satisfied \cite{zhu2023modeling}. 
Thus, the angle-of-arrival (AoA), angle-of-departure (AoD), and amplitude of the complex coefficient for each channel path received by users remain unchanged despite the movement of the MAs, which means that only the phases of multiple channel paths vary within the receive region \cite{zhu2023movable}. 

%We employ the field-response based channel model as in \cite{zhu2023modeling}, where the channel response is the superposition of the coefficients of multiple channel paths between the transceivers. We first describe the wiretap channel between the RIS and eavesdropper, and the legitimate channel between the RIS and $k$-th user, which are denoted as $\mathbf{h}_0$ and $\mathbf{h}_k$, $k \in \{1,2,\cdots,K\}$, respectively.
According to the field-response channel model \cite{zhu2023modeling,ma2023mimo}, we construct all the channels involved in Fig.~\ref{fig::systemModel}. First, we describe the channel vector between the RIS and eavesdropper, denoted by $\mathbf{h}_0$, and the channel vector between the RIS and $k$-th user, denoted by $\mathbf{h}_k$, $k \in \{1,2,\cdots,K\}$.
Denote the number of transmit and receive paths between the RIS and node $\kappa$ as $L^t_{\kappa}$ and $L^r_{\kappa}$, $\kappa \in  \{0,1,2,\cdots,K\}$, respectively. Note that $\kappa = 0$ represents the eavesdropper, and $\kappa = k$ for $k \in \{1,2,\cdots,K\}$ represents user $k$. Denote $\vartheta^r_{\kappa,i}$ and $\varphi^r_{\kappa,i}$ as the elevation and azimuth AoAs of the $i$-th receive path between the RIS and node $\kappa$, respectively. Subsequently, for node $\kappa$, the signal propagation difference $\rho^r_{\kappa,i}(\mathbf{u}_{\kappa})$ of the $i$-th receive path between the position of the FPA/MA $\mathbf{u}_{\kappa}$ and reference point $O_{\kappa} = [0,0]^{\rm T}$ can be represented as $\rho^r_{\kappa,i}(\mathbf{u}_{\kappa}) = x_{\kappa}\sin\vartheta^r_{\kappa,i}\cos\varphi^r_{\kappa,i} + y_{\kappa}\cos\vartheta^r_{\kappa,i}$ with $1 \leq i \leq L^r_{\kappa}$, and the corresponding phase difference is given by $\frac{2\pi}{\lambda}\rho^r_{\kappa,i}(\mathbf{u}_{\kappa})$, where $\lambda$ is the wavelength. Considering these phase differences across all $L^r_{\kappa}$ receive paths, the channel receive field-response vector (FRV) between the RIS and node $\kappa$ is represented as
\begin{equation} \label{eq::receive_FRV} \small
	\mathbf{f}_{\kappa}(\mathbf{u}_{\kappa}) \!=\! [e^{j\frac{2\pi}{\lambda}\rho^r_{\kappa,1}(\mathbf{u}_{\kappa})},  \!\cdots\!, e^{j\frac{2\pi}{\lambda}\rho^r_{\kappa,L^r_{\kappa}}(\mathbf{u}_{\kappa})}]^{\rm T}, \kappa \in \{0,1,\!\cdots\!,K\}.
\end{equation} 
Similarly, considering the phase differences across all $L^t_{\kappa}$ transmit paths between the $n$-th RIS reflection element and node $\kappa$, the transmit FRV, denoted by $\mathbf{g}_{\kappa}(\mathbf{t}_n)$, is expressed as
\begin{equation} \label{eq::transmit_FRV} \small
	\mathbf{g}_{\kappa}(\mathbf{t}_n) \!=\! [e^{j\frac{2\pi}{\lambda}\rho^t_{\kappa,1}(\mathbf{t}_n)}, \cdots,\! e^{j\frac{2\pi}{\lambda}\rho^t_{\kappa,L^t_{\kappa}}(\mathbf{t}_n)}]^{\rm T}, \kappa \in \{0,1,\cdots,K\},
\end{equation}
where $\rho^t_{\kappa,j}(\mathbf{t}_n) = \tilde{X}_n\sin\vartheta^t_{\kappa,j}\cos\varphi^t_{\kappa,j} + \tilde{Y}_n\cos\vartheta^t_{\kappa,j}$, $1 \leq j \leq L^t_{\kappa}$, denotes the difference in signal propagation distance of the $j$-th transmit channel path between $\mathbf{t}_n$ and the origin of the local coordinate system at the RIS. Therein, $\vartheta^t_{\kappa,j}$ and $\varphi^t_{\kappa,j}$ are the elevation and azimuth AoDs for the $j$-th transmit path between the RIS and node $\kappa$, respectively.

%Similarly, the transmit FRV from the $i$-th FPA of the BS to the $k$-th user is given by
%\begin{equation}
%	\mathbf{b}_{k,i} = [e^{j\frac{2\pi}{\lambda}\rho^t_{k,1}(\mathbf{r}_{B}_i)}, e^{j\frac{2\pi}{\lambda}\rho^t_{k,2}(\mathbf{r}_{B}_i)}, \cdots, e^{j\frac{2\pi}{\lambda}\rho^t_{k,L^t_k}(\mathbf{r}_{B}_i)}]^{\rm T},
%\end{equation}
%where $\rho^t_{k,j}(\mathbf{v}_n) = x_n\cos\theta^t_{k,j}\sin\phi^t_{k,j} + y_n\sin\theta^t_{k,j}$, $\theta^t_{k,j}$ and $\phi^t_{k,j}$ are the elevation and azimuth AoDs, respectively.

Furthermore, path-response matrix (PRM), denoted by $\mathbf{\Sigma}_{\kappa} \in \mathbb{C}^{L^r_{\kappa} \times L^t_{\kappa}}$, is defined to account for the responses between all channel paths from RIS reference point $\mathbf{t}_0 = [0, 0]^{\rm T}$ and the node $\kappa$ reference point $O_{\kappa}$. Therefore, the channel vector from the RIS to node $\kappa$ is obtained as
\begin{equation} \label{eq::channel_RIS_to_kappa} \small
	\mathbf{h}_{\kappa}(\mathbf{u}_{\kappa}) = (\mathbf{f}_{\kappa}(\mathbf{u}_{\kappa})^{\rm H}\mathbf{\Sigma}_{\kappa}\mathbf{G}_{\kappa})^{\rm T}, \kappa \in \{0,1,\cdots,K\},
\end{equation}
where $\mathbf{G}_{\kappa} = [\mathbf{g}_{\kappa}(\mathbf{t}_1),\mathbf{g}_{\kappa}(\mathbf{t}_2),\cdots,\mathbf{g}_{\kappa}(\mathbf{t}_N)]$ is the field-response matrix (FRM) at the RIS. 
Since the RIS and eavesdropper are equipped with FPAs, 
%in the following part of this paper,
$\mathbf{\Sigma}_{\kappa}$, $\mathbf{G}_{\kappa}$, and $\mathbf{f}_0(\mathbf{u}_0)$ for $\kappa \in \{0,1,\cdots,K\}$ are all constant, while $\mathbf{f}_k(\mathbf{u}_k)$ is a function changing with $\mathbf{u}_k$, $k \in \{1,\cdots,K\}$.
% 以上构建的是 MA-based 信道

% 以下是 UPA (FPA) 信道构建
Finally, we describe the channel matrix between the BS and RIS, denoted by $\mathbf{H}$. Similar to \eqref{eq::channel_RIS_to_kappa}, $\mathbf{H}$ can be expressed as
\begin{equation} \label{eq::channel_BS_to_RIS} \small
	\mathbf{H} = \mathbf{F}_s^{\rm H}\mathbf{\Sigma}_{bs}\mathbf{G}_b,
\end{equation}
where $\mathbf{F}_s = [\mathbf{f}_s(\mathbf{t}_1), \cdots, \mathbf{f}_s(\mathbf{t}_N)]$  denotes the receive FRM at the RIS, and $\mathbf{f}_s(\mathbf{t}_n) = [e^{j\frac{2\pi}{\lambda}\rho^r_{s,1}(\mathbf{t}_n)}, \cdots, e^{j\frac{2\pi}{\lambda}\rho^r_{s,L^r_{bs}}(\mathbf{t}_n)}]^{\rm T}$ with $1 \leq n \leq N$, is the receive FRV associated with the $n$-th element at the RIS. Besides, $\mathbf{G}_b = [\mathbf{g}_b(\mathbf{v}_1), \cdots, \mathbf{g}_b(\mathbf{v}_M)]$ denotes the transmit FRM at the BS, where $\mathbf{g}_b(\mathbf{v}_m) = [e^{j\frac{2\pi}{\lambda}\rho^t_{s,1}(\mathbf{v}_m)}, \cdots, e^{j\frac{2\pi}{\lambda}\rho^t_{s,L^t_{bs}}(\mathbf{v}_m)}]^{\rm T}$, $1 \leq m \leq M$, is the transmit FRV of the $m$-th element at the BS. Therein, denote $L^t_{bs}$ and $L^r_{bs}$ as the number of transmit and receive paths from the BS to RIS, respectively. In addition, $\vartheta^t_{s,j}$ and $\varphi^t_{s,j}$ as well as $\vartheta^r_{s,i}$ and $\varphi^r_{s,i}$ represent the elevation and azimuth AoDs for $j$-th transmit path as well as the elevation and azimuth AoAs for $i$-th receive path between the BS and RIS, respectively. As such, we have $\rho^t_{s,j}(\mathbf{v}_m) = X_m\sin\vartheta^t_{s,j}\cos\varphi^t_{s,j} + Y_m \cos\vartheta^t_{s,j}, 1 \leq j \leq L^t_{bs}$, and $\rho^r_{s,i}(\mathbf{t}_n) = \tilde{X}_n\sin\vartheta^r_{s,i}\cos\varphi^r_{s,i} + \tilde{Y}_n \cos\vartheta^r_{s,i}, 1 \leq i \leq L^r_{bs}$, respectively. Furthermore, $\mathbf{\Sigma}_{bs} \in \mathbb{C}^{L^r_{bs} \times L^t_{bs}}$ represents the PRM between all transmit and receive paths.

\subsection{Communication, Radar Sensing and Security Model}
%% Communication Model
Based on the above signal and channel models, the received signal at user $k$ can be obtained as
\begin{equation} \label{eq::received_signal} 
	y_k = \mathbf{g}_k^{\rm H}(\mathbf{\Phi},\mathbf{u}_k)\mathbf{x} + n_k, 
\end{equation}
where $\mathbf{g}_k^{\rm H}(\mathbf{\Phi},\mathbf{u}_k) = \mathbf{h}^{\rm H}_k(\mathbf{u}_k)\mathbf{\Phi}\mathbf{H} \in \mathbb{C}^{1 \times M}$ denotes the cascaded channel from the BS to user $k$ via the RIS, $k \in \{1,2,\cdots, K\}$, and
%$\mathbf{h}^{\rm H}_k(\mathbf{\Phi},\mathbf{u}_k) = \mathbf{h}^{\rm H}_{B,k}(\mathbf{u}_k) + \mathbf{h}^{\rm H}_{R,k}(\mathbf{u}_k)\mathbf{\Phi}\mathbf{H}_{BR} \in \mathbb{C}^{1 \times M}$ is the cascaded communication channel from the ISAC BS to user $k$. 
$\mathbf{\Phi} = \text{diag}\left\{\boldsymbol{\phi}\right\} \in \mathbb{C}^{N \times N}$ denotes the reflection-coefficient matrix of the RIS, with $\boldsymbol{\phi} \triangleq [\phi_1, \cdots, \phi_N]$, where $\phi_n$ denotes the reflection coefficient of the $n$-th element satisfying $|\phi_n| = 1$, $1 \leq n \leq N$.  
Besides, $n_k \sim \mathcal{CN}(0,\sigma^2_k)$ represents the additive white Gaussian noise (AWGN) at user $k$. 
It is assumed that the CSI of all relevant channels, i.e., $\mathbf{H}$ and $\mathbf{h}^{\rm H}_k(\mathbf{u}_k)$, is known at the BS through the two-timescale channel estimation approach \cite{hua2023secure}.
As a result, the received SINR for user $k$ is expressed as
% 替换为将w_c和w_r统一化
%\begin{equation} \small \label{eq::metric_communication}
%	\gamma_k \!=\! \frac{\left|\mathbf{h}^{\rm H}_k(\mathbf{u}_k)\mathbf{\Phi}\mathbf{H}\mathbf{w}_{c,k}\right|^2}{\sum\limits_{j=1,j \neq k}^{K}\!\left|\mathbf{h}^{\rm H}_k(\mathbf{u}_k)\mathbf{\Phi}\mathbf{H}\mathbf{w}_{c,j}\right|^2 \!\!+\!\! \sum\limits_{m=1}^{M}\!\left|\mathbf{h}^{\rm H}_k(\mathbf{u}_k)\mathbf{\Phi}\mathbf{H}\mathbf{w}_{r,m}\right|^2 \!\!+\! \sigma_k^2}.
%\end{equation}
\begin{equation} \label{eq::metric_communication} 
	\gamma_k = \frac{\left|\mathbf{g}_k^{\rm H}(\mathbf{\Phi},\mathbf{u}_k)\mathbf{w}_k\right|^2}{\sum\limits_{j=1,j \neq k}^{K+M}\left|\mathbf{g}_k^{\rm H}(\mathbf{\Phi},\mathbf{u}_k)\mathbf{w}_j\right|^2 + \sigma_k^2},
\end{equation}
where $\mathbf{w}_j$ denotes the $j$-th column in $\mathbf{W}$. Thus, $\mathbf{W}$ can be represented as $\mathbf{W} = [\mathbf{w}_1, \cdots, \mathbf{w}_j, \cdots, \mathbf{w}_{K+M}]$.

%% Radar Sensing Model
In addition to transmitting communication signals from the BS to user, $\mathbf{x}$ is also employed to detect the potential eavesdropping target.
To establish an effective BS-RIS-target link for sensing, the RIS is utilized to generate a virtual LoS link, thereby it is assumed that the BS has precise estimates of the eavesdropper's channel \cite{liu2022joint}.
%The RIS is utilized to create a virtual LoS link between the BS and the target, and thus forming an effective BS-RIS-target link for sensing. As such, it is assumed that  through efficient sensing.
% 与系统模型第一段重复
%As shown in Fig. \ref{fig::systemModel}, we consider the scenario where the direct paths between the BS and the potential target location is not available, because eavesdroppers may deliberately hide by the trees and other structures to avoid being discovered. Therefore, the RIS is responsible for creating strong virtual line-of-sight (LoS) links to maintain the communication and sensing performances. 
The echo signal that experiences the BS-RIS-target-RIS-BS link and is collected at the BS for target detection can be represented by
\begin{equation} \label{eq::echo_signal} 
	\begin{aligned}
		y_r = \alpha_t\mathbf{H}^{\rm H}\mathbf{\Phi}^{\rm H}\mathbf{A}_{RE}\mathbf{\Phi}\mathbf{H} \mathbf{x} + \mathbf{n}_r,
	\end{aligned}
\end{equation}
% 放到后面说
%where $\mathbf{r}_{B}^{\rm H}$ denotes the unit-norm receive beamformer at the BS satisfying $\left\|\mathbf{r}_{B}^{\rm H}\right\|_2 = 1$.
where $\alpha_t$ represents the radar cross section (RCS) following $\alpha_t \sim \mathcal{CN}(0,\sigma_t^2)$. Besides, $\mathbf{A}_{RE} = \mathbf{h}_0^{\rm H}\mathbf{h}_0$ and $\mathbf{n}_r \sim \mathcal{CN}(0, \sigma_r^2\mathbf{I}_M)$ denote the target response matrix at the RIS and
% 直接用h_0 取代原始定义
%In addition, $\mathbf{A}_{RE} = \beta_S \mathbf{a}(\theta_S, \phi_S)\mathbf{a}^{\rm H}(\theta_S, \phi_S)$ is the target response matrix at the RIS, where $\beta_S$ represents the complex gain of the RIS-target-RIS path, $\theta_S$ and $\phi_S$ are the elevation and azimuth AoDs at the RIS toward the sensing direction, in which the target may possibly exist. $\mathbf{a}(\theta_S, \phi_S)$ is the steering vector, which is given by
%\begin{equation} \small \label{eq::steering_vector}
%	\begin{aligned}
%			\mathbf{a}(\theta_S, \phi_S) &\!=\! \left[1, e^{j\frac{2\pi d}{\lambda}\sin\theta_S\cos\phi_S},\cdots,e^{j\frac{2\pi d}{\lambda}(N_1-1)\sin\theta_S\cos\phi_S}\right]^{\rm T} \otimes \\
%			&\!\!\!\!\! \left[1, e^{j\frac{2\pi d}{\lambda}\cos\theta_S},\cdots,e^{j\frac{2\pi d}{\lambda}(N_2-1)\cos\theta_S}\right]^{\rm T}, 
%		\end{aligned}
%\end{equation}
%where $d$ denotes the antenna/element spacing at the RIS. 
the noise received at the BS, respectively. Since radar sensing involves the analysis of received echo signals across $L$ samples, these samples are combined in a matrix form as
\begin{equation} \label{eq::echo_signals_L} 
	\mathbf{Y}_r = \alpha_t\mathbf{H}_t(\mathbf{\Phi})\mathbf{WS} + \mathbf{N}_r,
\end{equation}
with $\mathbf{H}_t(\mathbf{\Phi}) \triangleq \mathbf{H}^{\rm H}\mathbf{\Phi}^{\rm H}\mathbf{A}_{RE}\mathbf{\Phi}\mathbf{H}$, where $\mathbf{S} \triangleq [\mathbf{s}(1), \mathbf{s}(2), \cdots, \mathbf{s}(L)]$ and $\mathbf{N}_r \triangleq [\mathbf{n}_r(1), \mathbf{n}_r(2), \cdots, \mathbf{n}_r(L)]$ are defined as the symbol and noise matrices, respectively. In order to improve the target detection performance, the received echo signals are processed by a matched filter, which is expressed as
\begin{equation} \label{eq::echo_signals_L_filter} 
	\tilde{\mathbf{Y}}_r = \alpha_t\mathbf{H}_t(\mathbf{\Phi})\mathbf{WS}\mathbf{S}^{\rm H} + \mathbf{N}_r\mathbf{S}^{\rm H}.
\end{equation} 
Then, by defining $\tilde{\mathbf{y}}_r \triangleq {\rm vec}(\tilde{\mathbf{Y}}_r) \in \mathbb{C}^{M(K+M)\times 1}$, $\tilde{\mathbf{w}} \triangleq {\rm vec}(\mathbf{W}) \in \mathbb{C}^{M(K+M)\times 1}$, and $\tilde{\mathbf{n}}_r \triangleq {\rm vec}(\mathbf{N}_r\mathbf{S}^{\rm H}) \in \mathbb{C}^{M(K+M)\times 1}$, $\tilde{\mathbf{Y}}_r$ is vectorized as 
\begin{equation} \label{eq::echo_signals_L_filter_vectorized}
	\tilde{\mathbf{y}}_r = \alpha_t\left(\mathbf{S}\mathbf{S}^{\rm H} \otimes \mathbf{H}_t(\mathbf{\Phi})\right)\tilde{\mathbf{w}} + \tilde{\mathbf{n}}_r,
\end{equation}
which is processed by the receive beamformers, denoted by $\mathbf{r}_B \in \mathbb{C}^{M(K+M)\times 1}$, yielding
\begin{equation} \label{eq::echo_signals_L_filter_vectorized_receivedBeamformer}
	\mathbf{r}_B^{\rm H}\tilde{\mathbf{y}}_r = \alpha_t\mathbf{r}_B^{\rm H}\left(\mathbf{S}\mathbf{S}^{\rm H} \otimes \mathbf{H}_t(\mathbf{\Phi})\right)\tilde{\mathbf{w}} + \mathbf{r}_B^{\rm H}\tilde{\mathbf{n}}_r.
\end{equation}
% 雷达感知应当基于L个抽样， 因此修改为上述段落
As such, the radar output SNR is given by \cite{liu2022joint}
\begin{equation} \label{eq::metric_radar_SNR} 
	S_r = \frac{\sigma_t^2\mathbb{E}_{\mathbf{S}}\left\{\left|\mathbf{r}_B^{\rm H}\left(\mathbf{S}\mathbf{S}^{\rm H} \otimes \mathbf{H}_t(\mathbf{\Phi})\right)\tilde{\mathbf{w}}\right|^2\right\}}{L\sigma_r^2\mathbf{r}_{B}^{\rm H}\mathbf{r}_{B}}.
\end{equation}
According to \cite{liu2022joint}, $S_r$ is closely related to several radar performance metrics, including parameter estimation accuracy and detection probability, thereby it serves as a crucial performance indicator for target sensing.

%% Security Model
%We consider the scenario where the direct paths between the BS and the potential target location is not available due to the blockages, so that the RIS is responsible for creating strong virtual line-of-sight (LoS) links to maintain the communication and sensing performances. 

As the target acts as a potential eavesdropper, it attempts to intercept signals of legitimate links. As such, the received signals at the target can be expressed as
\begin{equation} \label{eq::eavesdropping_signal}
	y_e = \mathbf{g}_0^{\rm H}(\mathbf{\Phi})\mathbf{x} + n_e,
\end{equation}
where $\mathbf{g}_0^{\rm H}(\mathbf{\Phi}) = \mathbf{h}^{\rm H}_0\mathbf{\Phi}\mathbf{H} \in \mathbb{C}^{1 \times M}$ denotes the cascaded channel between the BS and the target, and $n_e \sim \mathcal{CN}(0,\sigma^2_e)$ represents the AWGN at the receiver. 
% 见前述，统一构建4条基本信道链路
%Let $\vartheta$ and $\varphi$ denote the azimuth and elevation AoDs from the RIS to the target, respectively. Accordingly, the steering vector from the RIS to the target towards direction $(\vartheta, \varphi)$ can be expressed as
%\begin{equation}\small \label{eq::steering_vector}
%	\begin{aligned}
%		\mathbf{a}(\vartheta, \varphi) &= \left[1, e^{-j\frac{2\pi d}{\lambda}\sin\vartheta\cos\varphi},\cdots,e^{-j\frac{2\pi (N_x-1)d}{\lambda}\sin\vartheta\cos\varphi}\right]\\&\otimes\left[1, e^{-j\frac{2\pi d}{\lambda}\sin\vartheta\sin\varphi},\cdots,e^{-j\frac{2\pi (N_y-1)d}{\lambda}\sin\vartheta\sin\varphi}\right],
%	\end{aligned}
%\end{equation}
%where $d$ denotes the spacing between two adjacent reflecting elements. Thus, we have $\mathbf{h}_{RE}^{\rm H} = \alpha_{RE} \mathbf{a}(\vartheta, \varphi)$, where $\alpha_{RE}$ represents the large-scale fading coefficient. 
Thus, the eavesdropping SINR for the target to intercept the information intended for user $k$ can be obtained as
% 将w_c和w_r统一化
%\begin{equation} \small \label{eq::metric_eavesdropping}
%	\gamma_{e,k} = \frac{|\mathbf{g}^{\rm H}(\mathbf{\Phi})\mathbf{w}_{c,k}|^2}{\sum\limits_{j=1,j \neq k}^{K}|\mathbf{g}^{\rm H}(\mathbf{\Phi})\mathbf{w}_{c,j}|^2 + \sum\limits_{m=1}^{M}|\mathbf{g}^{\rm H}(\mathbf{\Phi})\mathbf{w}_{r,m}|^2 + \sigma_e^2}.
%\end{equation}
\begin{equation} \label{eq::metric_eavesdropping} 
	\gamma_{e,k} = \frac{|\mathbf{g}_0^{\rm H}(\mathbf{\Phi})\mathbf{w}_k|^2}{\sum\limits_{j=1,j \neq k}^{K+M}|\mathbf{g}_0^{\rm H}(\mathbf{\Phi})\mathbf{w}_j|^2 + \sigma_e^2}.
\end{equation}

\subsection{Problem Formulation}
In the considered MAs-aided RIS-ISAC system, each user may have heterogeneous secrecy requirements. Thus, considering secrecy leakage SINR constraints leads to a more flexible resource allocation (e.g. transmit power and beamformer) compared to imposing constraints on the secrecy rate or directly optimizing it. In this paper, to maximize the sum-rate for users, we jointly optimize the transmit beamformer $\mathbf{W}$, RIS reflection coefficient matrix $\mathbf{\Phi}$, receive filter $\mathbf{r}_{B}$, and positions of MAs $\{\mathbf{u}_k\}_{k=1}^{K}$, while satisfying the minimum radar output SNR, the minimum required communication SINR, and the maximum allowable secrecy leakage. The secure transmission optimization problem for RIS-ISAC systems is thus formulated as 

\vspace{-1 em}
\begin{subequations}  \label{eq::originalFormulation}
	\begin{align}
		 \max_{\mathbf{W},\mathbf{\Phi},\mathbf{r}_{B},\{\mathbf{u}_k\}} & \sum_{k=1}^{K} \log_2(1+\gamma_k) \label{eq::originalFormulation_objective}\\
		 {\rm s.t.}~ & S_r \geq \Gamma_r, \label{eq::originalFormulation_b}\\
		 &\gamma_k \geq \Gamma_k, 1 \leq k \leq K, \label{eq::originalFormulation_c}\\
		 & \gamma_{e,k} \leq \Gamma_{e,k}, 1 \leq k \leq K, \label{eq::originalFormulation_d}\\
		 & \left\|\mathbf{W}\right\|_F \leq P_B, \label{eq::originalFormulation_e}\\
%		 & \left\|\mathbf{r}_{B}^{\rm H}\right\|_2 = 1,
%		 \label{eq::originalFormulation_f}\\
		 & \left|\phi_n\right| = 1, 1 \leq n \leq N, \label{eq::originalFormulation_g} \\
		 & \mathbf{u}_k \in \mathcal{C}_k, 1 \leq k \leq K, \label{eq::originalFormulation_h}
	\end{align}
\end{subequations}
where the radar SNR requirement is specified in constraint \eqref{eq::originalFormulation_b}, with $\Gamma_r$ representing the minimum SNR required for the echo signal at the BS; constraint \eqref{eq::originalFormulation_c} ensures that the SINR at user $k$ is no less than the predefined threshold $\Gamma_k$; constraint \eqref{eq::originalFormulation_d} limits the maximum secrecy leakage for $k$-th user at the eavesdropping target to below $\Gamma_{e,k}$; $P_B$ in constraint \eqref{eq::originalFormulation_e} denotes the transmit power budget at the BS; constraint \eqref{eq::originalFormulation_g} imposes a unit-modulus constraint on the RIS elements reflection coefficients; and constraint \eqref{eq::originalFormulation_h} indicates that each MA can only move within the given receive region, i.e., $\mathcal{C}_k$. It is noted that with constraints \eqref{eq::originalFormulation_c} and \eqref{eq::originalFormulation_d}, the level of PLS for the considered system can be guaranteed. Specifically, the secrecy rate for user $k$, denoted by $R_k$, is lower-bounded by
%$\max\{\log_2(1+\Gamma_k)-\log_2(1+\Gamma_{e,k}),0\}$.
\begin{equation} \small \label{eq::PLS_level}
	R_k \geq \max\{\log_2(1+\Gamma_k)-\log_2(1+\Gamma_{e,k}),0\}.
\end{equation}
Due to the complicated objective \eqref{eq::originalFormulation_objective} in fractional form, problem \eqref{eq::originalFormulation} is highly non-convex. Besides, the coupling between  $\mathbf{W},\mathbf{\Phi},\mathbf{r}_{B}$, and $\{\mathbf{u}_k\}$ in both the objective function and constraints makes \eqref{eq::originalFormulation} more intractable.

\section{Proposed Solution}
In this section, we first transform \eqref{eq::originalFormulation_objective} into a  tractable expression by implementing the fractional programming (FP) technique. Then, a two-layer penalty-based algorithm is proposed to decouple optimization variables.

%in different blocks, which includes two-layer iterations, i.e., the inner layer and outer layer. The inner layer solves the penalized optimization problem in an alternate manner. Specifically, five subproblems are solved in the sequel, which respectively optimize the auxiliary variable set $\Omega$, receive filter $\mathbf{r}_{B}$, transmit beamformer $\mathbf{W}$, RIS reflection coefficients matrix $\mathbf{\Phi}$, and MA positions $\{\mathbf{u}_k\}$, with all the other variables being fixed. The outer layer updates the penalty factor over iterations to guarantee convergence.

\subsection{Problem Transformation}
%\subsubsection{FP-based Transformation for objective in \eqref{eq::originalFormulation_objective}}
By exploiting the FP technique, \eqref{eq::originalFormulation_objective} can be transformed into a polynomial expression 
%\begin{lemma}
%	By resorting to the FP approach in [x], \eqref{eq::originalFormulation_objective} can be expressed as
%	\begin{equation} \small
%		\begin{aligned}
%			&f(\mathbf{W},\mathbf{\Phi},\{\mathbf{u}_k\},\boldsymbol{\Lambda},\mathbf{c}) \triangleq \sum_{k=1}^{K}\log_2(1+\lambda_k) - \sum_{k=1}^{K}\lambda_k - \sum_{k=1}^{K}|c_k|^2\sigma_k^2 \\
%			& \!\!-\!\! \sum_{k=1}^{K}|c_k|^2\!\!\!\!\!\!\sum\limits_{j=1,j \neq k}^{K+M}\!\!\!\!\left|\mathbf{g}_k^{\rm H}(\mathbf{\Phi},\mathbf{u}_k)\mathbf{w}_j\right|^2 \!\!\!+\! 2\sum_{k=1}^{K}\!\!\sqrt{1+\lambda_k}\Re\{c_k^*\mathbf{g}_k^{\rm H}(\mathbf{\Phi},\mathbf{u}_k)\mathbf{w}_k\}.
%		\end{aligned}
%	\end{equation}
%\end{lemma}
via introducing auxiliary variables $\boldsymbol{\lambda} = [\lambda_1, \cdots, \lambda_K]$, and the Lagrangian dual expression of \eqref{eq::originalFormulation_objective} can be obtained as
\begin{equation} \label{eq::FP_1} \small
	\sum_{k=1}^{K}\!\log_2(1+\lambda_k) - \sum_{k=1}^{K}\!\lambda_k + \sum_{k=1}^{K}\!\frac{(1+\lambda_k)\!\left|\mathbf{g}_k^{\rm H}(\mathbf{\Phi},\mathbf{u}_k)\mathbf{w}_k\right|^2}{\sum\limits_{j=1,j \neq k}^{K+M}\!\left|\mathbf{g}_k^{\rm H}(\mathbf{\Phi},\mathbf{u}_k)\mathbf{w}_j\right|^2 \!+\! \sigma_k^2},
\end{equation} 
where $\mathbf{W}$, $\mathbf{\Phi}$, and $\{\mathbf{u}_k\}$ are extracted from the $\log(\cdot)$ function, yet they remain coupled in the third term. Since the third term is still in a fractional form, we introduce auxiliary variables $\boldsymbol{\iota} = [\iota_1,\cdots,\iota_K]$ and re-invoke FP to further transform \eqref{eq::FP_1} into a polynomial expression, which yields
\begin{equation} \label{eq::FP_2} \small
	\begin{aligned}
		&f(\mathbf{W},\mathbf{\Phi},\{\mathbf{u}_k\},\boldsymbol{\lambda},\boldsymbol{\iota}) \triangleq
		 \sum_{k=1}^{K}\log_2(1+\lambda_k) - \sum_{k=1}^{K}\lambda_k - \sum_{k=1}^{K}|\iota_k|^2\sigma_k^2 \\
		& \!\!-\!\! \sum_{k=1}^{K}|\iota_k|^2\!\!\!\!\!\!\sum\limits_{j=1,j \neq k}^{K+M}\!\!\!\!\left|\mathbf{g}_k^{\rm H}(\mathbf{\Phi},\mathbf{u}_k)\mathbf{w}_j\right|^2 \!\!\!+\! 2\sum_{k=1}^{K}\!\!\sqrt{1+\lambda_k}\Re\{\iota_k^*\mathbf{g}_k^{\rm H}(\mathbf{\Phi},\mathbf{u}_k)\mathbf{w}_k\}.
	\end{aligned}
\end{equation}

%Thus, we have the following optimization problem
%\begin{subequations}
%	\begin{align}
%		&\max_{\mathbf{W},\mathbf{\Phi},\mathbf{r}_{B},\{\mathbf{u}_k\},\boldsymbol{\Lambda},\mathbf{c}}  f(\mathbf{W},\mathbf{\Phi},\{\mathbf{u}_k\},\boldsymbol{\Lambda},\mathbf{c}) \\
%		&{\rm s.t.}~  \eqref{eq::originalFormulation_b}, \eqref{eq::originalFormulation_c}, \eqref{eq::originalFormulation_d}, \eqref{eq::originalFormulation_e},\eqref{eq::originalFormulation_f},\eqref{eq::originalFormulation_g}, \eqref{eq::originalFormulation_h},\eqref{eq::originalFormulation_i} 
%	\end{align}
%\end{subequations}

%\subsubsection{Problem Reformulation}
Subsequently, to separate the optimization variables into distinct blocks, a penalty-based algorithm is developed. Specifically, we introduce several auxiliary variables $ \left\{x_j,  z_{k,j},  1 \leq k \leq K, 1 \leq j \leq K+M\right\}$. Then, let $\mathbf{g}_0^{\rm H}(\mathbf{\Phi})\mathbf{w}_j = x_j$ and $\mathbf{g}^{\rm H}_k(\mathbf{\Phi},\mathbf{u}_k)\mathbf{w}_j = z_{k,j}$, respectively. Thus, problem \eqref{eq::originalFormulation} can be equivalently transformed into

\vspace{-1 em}
\begin{subequations}\label{eq::ReFormulation} \small
	\begin{align}
		&\!\!\!\!\!\!\!\!\!\!\!\!\!\!\!\!\!\max_{\mathbf{W},\mathbf{\Phi},\mathbf{r}_{B},\{\mathbf{u}_k\},\Omega} \quad \sum_{k=1}^{K}\log_2(1+\lambda_k) - \sum_{k=1}^{K}\lambda_k - \sum_{k=1}^{K}|\iota_k|^2\sigma_k^2 \nonumber \\ 
		&\!\!\!\!\!\!\!\!\!\!-\sum_{k=1}^{K}|\iota_k|^2\!\!\!\!\sum\limits_{j=1,j \neq k}^{K+M}\left|z_{k,j}\right|^2 + 2\sum_{k=1}^{K}\sqrt{1+\lambda_k}\Re\{\iota_k^*z_{k,k}\}   \label{eq::ReFormulation_objective}\\
		{\rm s.t.}~  &S_r \geq \Gamma_r, \label{eq::ReFormulation_b}\\ &\frac{\left|z_{k,k}\right|^2}{\sum\limits_{j = 1, j \neq k}^{K+M}\!\left|z_{k,j}\right|^2 + \sigma_k^2} \geq \Gamma_k, 1 \leq k \leq K, \label{eq::ReFormulation_c}\\
		& \frac{|x_k|^2}{\sum\limits_{j = 1, j \neq k}^{K+M}|x_j|^2 + \sigma_e^2} \leq \Gamma_{e,k}, 1 \leq k \leq K, \label{eq::ReFormulation_d}\\
		& \mathbf{g}_0^{\rm H}(\mathbf{\Phi})\mathbf{w}_j = x_j, \mathbf{g}^{\rm H}_k(\mathbf{\Phi},\mathbf{u}_k)\mathbf{w}_j = z_{k,j}, \nonumber \\ 
		& 1 \leq k \leq K, 1\leq j \leq K+M \label{eq::ReFormulation_e}\\
		& \eqref{eq::originalFormulation_e},\eqref{eq::originalFormulation_g}, \eqref{eq::originalFormulation_h}, \nonumber
	\end{align}
\end{subequations}
with $\Omega \triangleq \{\boldsymbol{\lambda}, \boldsymbol{\iota}, x_j,  z_{k,j}, 1 \leq k \leq K,  1 \leq j \leq K+M\}$.
Next, we add \eqref{eq::ReFormulation_e} to \eqref{eq::ReFormulation_objective} as penalty terms and reformulate the optimization problem as

\vspace{-1 em}
\begin{subequations}  \label{eq::ReFormulation_2} \small
	\begin{align}
		&\!\!\!\!\!\!\!\!\!\!\max_{\mathbf{W},\mathbf{\Phi},\mathbf{r}_{B},\{\mathbf{u}_k\},\Omega} \quad \sum_{k=1}^{K}\log_2(1+\lambda_k) - \sum_{k=1}^{K}\lambda_k - \sum_{k=1}^{K}|\iota_k|^2\sigma_k^2  \nonumber\\ 
		&\!\!\!\!\!\!\!\!\!\!-\sum_{k=1}^{K}|\iota_k|^2\!\!\!\!\sum\limits_{j=1,j \neq k}^{K+M}\!\!\!\!\! \left|z_{k,j}\right|^2 + 2\sum_{k=1}^{K}\!\!\!\sqrt{1+\lambda_k}\Re\{\iota_k^*z_{k,k}\} \!-\! \frac{1}{2\rho} \times \nonumber \\
		&\!\!\!\!\!\!\!\!\!\! \Big(\! \sum_{j=1}^{K+M}\!\!\left|\mathbf{g}_0^{\rm H}(\mathbf{\Phi})\mathbf{w}_{j} \!-\! x_{j}\right|^2  \!\!+\!\! \sum_{k=1}^{K}\sum_{j=1}^{K+M}\!\!\left|\mathbf{g}^{\rm H}_k(\mathbf{\Phi},\mathbf{u}_k)\mathbf{w}_j \!-\! z_{k,j}\right|^2 \!\Big) \label{eq::ReFormulation_2_objective}\\
		{\rm s.t.}~ &  \eqref{eq::originalFormulation_e},\eqref{eq::originalFormulation_g},\eqref{eq::originalFormulation_h},\eqref{eq::ReFormulation_b},\eqref{eq::ReFormulation_c},\eqref{eq::ReFormulation_d},  \nonumber
	\end{align}
\end{subequations}
where $\rho > 0$ denotes the penalty factor applied to penalize deviations from the equality constraints in \eqref{eq::ReFormulation_e}. 
%可加回
%It is noted that even though the equality constraints are relaxed in \eqref{eq::ReFormulation_2}, as $\rho \rightarrow 0$, the solution derived from solving \eqref{eq::ReFormulation_2} typically adheres to the constraints in \eqref{eq::ReFormulation_e}. 
To solve problem \eqref{eq::ReFormulation_2}, we introduce a two-layer iterative algorithm, where the penalty parameter is gradually updated in the outer layer. In the inner layer, we alternatively optimize the variables given any fixed penalty factor.

\subsection{Inner Layer Optimization}
We partition all the variables in \eqref{eq::ReFormulation_2} into five distinct blocks in the inner-layer iterations, including the auxiliary variable set $\Omega$, receive filter $\mathbf{r}_B$, transmit beamformer $\mathbf{W}$, RIS coefficient matrix $\mathbf{\Phi}$, and positions of MA $\{\mathbf{u}_k\}_{k=1}^{K}$. Then, we alternately  optimize each of them with all the other variables being fixed.
\subsubsection{Optimization of $\Omega$ with given $\mathbf{r}_B$, $\mathbf{W}$, $\mathbf{\Phi}$, and $\{\mathbf{u}_k\}_{k=1}^{K}$}
The auxiliary variable optimization subproblem can be expressed as

\vspace{-1 em}
\begin{subequations}  \label{eq::subproblem_Omega} \small
	\begin{align}
		&\!\!\!\!\!\!\!\!\!\!\max_{\Omega} \quad \sum_{k=1}^{K}\log_2(1+\lambda_k) - \sum_{k=1}^{K}\lambda_k - \sum_{k=1}^{K}|\iota_k|^2\sigma_k^2  \nonumber\\ 
		&\!\!\!\!\!\!\!\!\!\!-\sum_{k=1}^{K}|\iota_k|^2\!\!\!\!\sum\limits_{j=1,j \neq k}^{K+M}\!\!\!\!\! \left|z_{k,j}\right|^2 + 2\sum_{k=1}^{K}\!\!\!\sqrt{1+\lambda_k}\Re\{\iota_k^*z_{k,k}\} \!-\! \frac{1}{2\rho} \times \nonumber\\
		&\!\!\!\!\!\!\!\!\!\! \Big(\! \sum_{j=1}^{K+M}\!\!\left|\mathbf{g}_0^{\rm H}(\mathbf{\Phi})\mathbf{w}_{j} \!-\! x_{j}\right|^2  \!\!+\!\! \sum_{k=1}^{K}\sum_{j=1}^{K+M}\!\!\left|\mathbf{g}^{\rm H}_k(\mathbf{\Phi},\mathbf{u}_k)\mathbf{w}_j \!-\! z_{k,j}\right|^2 \!\Big) \label{eq::subproblem_Omega_objective}\\
		{\rm s.t.}~ &  \eqref{eq::ReFormulation_c},\eqref{eq::ReFormulation_d}. \nonumber
	\end{align}
\end{subequations}
For given the other variables, optimizing $\boldsymbol{\lambda}$ transforms into an unconstrained convex problem. As such, we can obtain its optimal solution by taking the gradient of  \eqref{eq::subproblem_Omega_objective} w.r.t. $\boldsymbol{\lambda}$ and equating it to zero, resulting in
\begin{equation} \label{eq::optimal_lambda} \small
	\lambda_k^{\rm opt} = \frac{\left|\mathbf{g}_k^{\rm H}(\mathbf{\Phi},\mathbf{u}_k)\mathbf{w}_k\right|^2}{\sum_{j=1,j \neq k}^{K+M}\left|\mathbf{g}_k^{\rm H}(\mathbf{\Phi},\mathbf{u}_k)\mathbf{w}_j\right|^2 + \sigma_k^2}, 1 \leq k \leq K.
\end{equation}
Similarly, the optimal $\iota_k^{\rm opt}$ is given by
\begin{equation} \label{eq::optimal_c} \small
	\iota_k^{\rm opt} = \frac{\sqrt{1+\lambda_k}\mathbf{g}_k^{\rm H}(\mathbf{\Phi},\mathbf{u}_k)\mathbf{w}_k}{\sum_{j=1,j \neq k}^{K+M}\left|\mathbf{g}_k^{\rm H}(\mathbf{\Phi},\mathbf{u}_k)\mathbf{w}_j\right|^2 + \sigma_k^2}, 1 \leq k \leq K.
\end{equation}

%\subsubsection{Auxiliary Variable Set $\Omega$ Optimization}

Then, it is noted that optimization variables corresponding to different blocks, i.e., $\{z_{k,j}, \forall j\}$ for $1 \leq k \leq K$ and $\{x_j, \forall j\}$, are independent in \eqref{eq::subproblem_Omega_objective}, \eqref{eq::ReFormulation_c}, and \eqref{eq::ReFormulation_d}. As such, problem \eqref{eq::subproblem_Omega} can be further decomposed into $K+1$ subproblems, which can be addressed in parallel. On the one hand, the subproblem associated with the $k$-th block $\{z_{k,j}, \forall j\}$ by ignoring irrelevant constants w.r.t. $\{z_{j,k}\}$ in \eqref{eq::subproblem_Omega} is expressed as

\vspace{-1 em}
\begin{subequations}  \label{eq::subproblem_Omega_z} \small
	\begin{align}
		\min_{\{z_{k,j}\}}& \quad  |\iota_k|^2\sum\limits_{j=1,j \neq k}^{K+M}\left|z_{k,j}\right|^2 - 2\sqrt{1+\lambda_k}\Re\{\iota_k^*z_{k,k}\} \nonumber\\ & \quad\quad + \frac{1}{2\rho} 
		 \sum_{j=1}^{K+M}\left|\mathbf{g}^{\rm H}_k(\mathbf{\Phi},\mathbf{u}_k)\mathbf{w}_j - z_{k,j}\right|^2 \label{eq::subproblem_Omega_z_objective}\\
		{\rm s.t.}&~  \quad \eqref{eq::ReFormulation_c}. \nonumber
	\end{align}
\end{subequations}
It is observed that  \eqref{eq::subproblem_Omega_z_objective} is convex since we have $|\iota_k|^2 + \frac{1}{2\rho} > 0$, despite the non-convex constraint. As demonstrated in \cite{bertsekas2016nonlinear}, strong duality is applicable to any optimization problem with a quadratic objective and a single quadratic inequality constraint, provided Slater’s condition is met. Hence, problem \eqref{eq::subproblem_Omega_z} can be addressed via addressing its dual problem. Specifically, by introducing a dual variable $\mu_{1,k} \geq 0$ corresponding to constraint \eqref{eq::ReFormulation_c}, the Lagrangian function for problem \eqref{eq::subproblem_Omega_z} is expressed as
\begin{equation} \label{eq::subproblem_Omega_z_Lagrangian} \small
	\begin{aligned}
			\mathcal{L}_1(z_{k,j},\mu_1) &= |\iota_k|^2\sum\limits_{j=1,j \neq k}^{K+M}\left|z_{k,j}\right|^2 - 2\sqrt{1+\lambda_k}\Re\{\iota_k^*z_{k,k}\} \\ &\!\!\!\!\!\!\!\!\!\!\!\! + \frac{1}{2\rho} 
			\sum_{j=1}^{K+M}\left|\mathbf{g}^{\rm H}_k(\mathbf{\Phi},\mathbf{u}_k)\mathbf{w}_j - z_{k,j}\right|^2 \\  
			&\!\!\!\!\!\!\!\!\!\!\!\!+ \mu_{1,k} \times \bigg(\Gamma_k\Big(\sum\limits_{j = 1, j \neq k}^{K+M}\!\left|z_{k,j}\right|^2 + \sigma_k^2\Big) - \left|z_{k,k}\right|^2\bigg).
	\end{aligned}
\end{equation}
Consequently, the associated dual function is represented by $\min_{\{z_{k,j}\}}\mathcal{L}_1(z_{k,j},\mu_{1,k})$. To ensure that the dual function remains bounded, it can be readily verified that $0 \leq \mu_{1,k} < 1$. Then, by differentiating $\mathcal{L}_1(z_{k,j},\mu_{1,k})$ w.r.t. $z_{k,j}$ and equating the derivative to zero, we obtain the optimal solution as
\begin{equation} \label{optimal_z} \small
	\begin{aligned} 
		z_{k,j}^{\rm opt}(\mu_{1,k}) \!\!=\!\! \left\{ \begin{array}{l} \!\!\!\!\!
			\frac{\mathbf{g}^{\rm H}_k(\mathbf{\Phi},\mathbf{u}_k)\mathbf{w}_j}{4\iota_k\rho + 1 + 2\mu_{1,k}\Gamma_k\rho}, 1 \leq j \leq K+M, j \neq k,\\
			\!\!\!\!\! \frac{\mathbf{g}^{\rm H}_k(\mathbf{\Phi},\mathbf{u}_k)\mathbf{w}_k + 2\rho\sqrt{1+\lambda_k}}{1 - 2\mu_{1,k}\rho},j = k.
		\end{array} \right.
	\end{aligned}
\end{equation}
Recall that for the optimal solution $z_{k,j}^{\rm opt}(\mu_{1,k})$ and $\mu_{1,k}^{\rm opt}$, the  complementary slackness condition must be
satisfied as follows \cite{boyd2004convex}
\begin{equation} \label{eq::complementary_slackness}  \small
	\mu_{1,k}^{\rm opt} \bigg(\Gamma_k\Big(\sum\limits_{j = 1, j \neq k}^{K+M}\!\left|z_{k,j}^{\rm opt}(\mu_{1,k})\right|^2 + \sigma_k^2\Big) - \left|z_{k,k}^{\rm opt}(\mu_{1,k})\right|^2\bigg) = 0.
\end{equation}
Next, we verify whether $\mu_{1,k}^{\rm opt} = 0$ is indeed the optimal solution. If constraint \eqref{eq::ReFormulation_c} is not satisfied with equality at the optimal solution, i.e., $\mu_{1,k}^{\rm opt} = 0$, then the optimal solution to problem \eqref{eq::subproblem_Omega_z} is obtained as $z_{k,j}^{\rm opt}(0)$. Otherwise, the optimal $\mu_{1,k}^{\rm opt}$ will be a positive value ($0 < \mu_{1,k}^{\rm opt} < 1$) that satisfies \eqref{eq::ReFormulation_c}, i.e.,
\begin{equation} \small \label{eq::positive_mu}
	\Gamma_k\Big(\sum\limits_{j = 1, j \neq k}^{K+M}\!\left|z_{k,j}^{\rm opt}(\mu_{1,k})\right|^2 + \sigma_k^2\Big) - \left|z_{k,k}^{\rm opt}(\mu_{1,k})\right|^2 = 0.
\end{equation}
In addition, it can be checked that $|z_{k,j}^{\rm opt}(\mu_{1,k})|$ for $j \neq k$ decreases monotonically with $\mu_{1,k}$, whereas  $|z_{k,k}^{\rm opt}(\mu_{1,k})|$ increases monotonically with $\mu_{1,k}$ for $0 < \mu_{1,k} < 1$. Consequently,  $\mu_{1,k}^{\rm opt}$ can be determined using a simple bisection search method within the interval $[0,1]$.

On the other hand, the subproblem w.r.t. block $\{x_j, \forall j\}$ can be expressed as

\vspace{-1 em}
\begin{subequations}  \label{eq::subproblem_Omega_x} \small
	\begin{align}
		\min_{\{x_j\}}& \quad   
		 \sum_{j=1}^{K+M}|\mathbf{g}_0^{\rm H}(\mathbf{\Phi})\mathbf{w}_{j} - x_{j}|^2 \\
		  {\rm s.t.}&~  \quad \eqref{eq::ReFormulation_d},  \nonumber
	\end{align}
\end{subequations}
which is a quadratically constrained quadratic program problem. Following \cite{bertsekas2016nonlinear}, strong duality applies to problem \eqref{eq::subproblem_Omega_x}. Therefore, problem \eqref{eq::subproblem_Omega_x} can be solved similarly to \eqref{eq::subproblem_Omega_z}. Specifically, by introducing dual variable $\mu_2$ ($\mu_2 \geq 0$) corresponding to constraint \eqref{eq::ReFormulation_d}, the Lagrangian function of \eqref{eq::subproblem_Omega_x} is obtained as
\begin{equation} \label{subproblem_Omega_x_Lagrangian} \small
	\begin{aligned}
		\mathcal{L}_2(x_j,\mu_2) &= \sum_{j=1}^{K+M}|\mathbf{g}_0^{\rm H}(\mathbf{\Phi})\mathbf{w}_{j} - x_{j}|^2 \\ & \quad \, + \mu_2\bigg(|x_k|^2 - \Gamma_{e,k}\Big(\sum\limits_{j = 1, j \neq k}^{K+M}|x_j|^2 + \sigma_e^2\Big)\bigg).
	\end{aligned}
\end{equation}
Thus, $\min_{\{x_j\}}\mathcal{L}_2(x_j,\mu_2)$ is the dual function for problem \eqref{eq::subproblem_Omega_x}.
% 如果需要，可以写成引理的形式
Note that we must have $1 \leq \mu_2 < 1$ to guarantee the dual function is bounded. Then, by applying the first-order optimality conditions, the optimal solution of the dual function is given by
\begin{equation} \label{eq::optimal_x} \small
	\begin{aligned}
		&x^{\rm opt}_j(\mu_2) = \left\{\!\! \begin{array}{l}
			\frac{\mathbf{g}^{\rm H}_0(\mathbf{\Phi})\mathbf{w}_j}{1-\mu_2\Gamma_{e,k}}, 1 \leq j \leq K+M, j \neq k,\\
			\frac{\mathbf{g}^{\rm H}_0(\mathbf{\Phi})\mathbf{w}_k}{1 + \mu_2}, j = k.
		\end{array} \right.
	\end{aligned}
\end{equation}
Similar to \eqref{eq::complementary_slackness}, the complementary slackness condition for the optimal dual variable $\mu_2^{\rm opt}$ is given by
\begin{equation} \small
	\mu_2^{\rm opt}\bigg(|x^{\rm opt}_k(\mu_2^{\rm opt})|^2 - \Gamma_{e,k}\Big(\sum\limits_{j = 1, j \neq k}^{K+M}|x^{\rm opt}_j(\mu_2^{\rm opt})|^2 + \sigma_e^2\Big)\bigg) = 0.
\end{equation}
%Similarly, we check whether $\mu_2^{\rm opt} = 0$ is the optimal solution or not. 
If constraint \eqref{eq::ReFormulation_d} is not satisfied with equality at the optimal solution, i.e., $\mu_2^{\rm opt} = 0$, the optimal solution for problem \eqref{eq::subproblem_Omega_x} is given by $x_j^{\rm opt}(0)$. Otherwise, similar to \eqref{eq::positive_mu},  $\mu_2^{\rm opt}$ is positive (i.e., $0 < \mu_2^{\rm opt} < 1$) that satisfies 
\begin{equation} \small
|x^{\rm opt}_k(\mu_2^{\rm opt})|^2 - \Gamma_{e,k}\bigg(\sum\limits_{j = 1, j \neq k}^{K+M}|x^{\rm opt}_j(\mu_2^{\rm opt})|^2 + \sigma_e^2\bigg)=0,
\end{equation}
and $\mu_2^{\rm opt}$ can be determined using a simple bisection search method within the interval $[0,1]$.

% 类同，精简删除
%Finally, it can be readily verified that $|x_j^{\rm opt}(\mu_2)|$ for $j \neq k$ is monotonically increasing w.r.t. $\mu_2$, while $|x_k^{\rm opt}(\mu_2)|$ is monotonically decreasing w.r.t. $\mu_2$ for $0 < \mu_2 < 1$. As such, the optimal $\mu_2^{\rm opt}$ can be obtained by applying a simple bisection search method between 0 and 1.

\subsubsection{Optimization of $\mathbf{r}_B$ with given $\Omega$, $\mathbf{W}$, $\mathbf{\Phi}$, and $\{\mathbf{u}_k\}_{k=1}^{K}$} In this subproblem, we first approximate the received echo signal SNR in \eqref{eq::metric_radar_SNR} by applying Jensen's inequality $\mathbb{E}\left\{f(x)\right\} \geq f\left(\mathbb{E}\{x\}\right)$, which is given by
\begin{equation} \label{eq::radar_SNR_lb} \small
	S_r \geq \frac{L\sigma_t^2 \left|\mathbf{r}_B^{\rm H}\left(\mathbf{I}_{K+M} \otimes \mathbf{H}_t(\mathbf{\Phi})\right)\tilde{\mathbf{w}}\right|^2}{\sigma_r^2\mathbf{r}_{B}^{\rm H}\mathbf{r}_{B}},
\end{equation}
with $\mathbb{E}\left\{\mathbf{S}\mathbf{S}^{\rm H}\right\} = L\mathbf{I}_{K+M}$,
%, and the right-hand side of \eqref{eq::radar_SNR_lb} represents the worst-case achieved radar SNR in \eqref{eq::originalFormulation_b}.
Thus, the subproblem for optimizing $\mathbf{r}_{B}$ is given by
\begin{equation} \small \label{eq::Subproblem_rBS}
		\max_{\mathbf{r}_{B}} \quad \frac{L\sigma_t^2 \left|\mathbf{r}_B^{\rm H}\left(\mathbf{I}_{K+M} \otimes \mathbf{H}_t(\mathbf{\Phi})\right)\tilde{\mathbf{w}}\right|^2}{\sigma_r^2\mathbf{r}_{B}^{\rm H}\mathbf{r}_{B}}.
\end{equation}
%\vspace{-1 em}
%\begin{subequations} \small \label{eq::Subproblem_rBS}
%	\begin{align}
%		\max_{\mathbf{r}_{B}}& \quad \frac{L\sigma_t^2 \left|\mathbf{r}_B^{\rm H}\left(\mathbf{I}_{K+M} \otimes \mathbf{H}_t(\mathbf{\Phi})\right)\tilde{\mathbf{w}}\right|^2}{\sigma_r^2\mathbf{r}_{B}^{\rm H}\mathbf{r}_{B}} \\
%		{\rm s.t.}&~ \quad  \eqref{eq::originalFormulation_f}, \nonumber
%	\end{align}
%\end{subequations}
Problem \eqref{eq::Subproblem_rBS} is a Rayleigh quotient, and we can determine its optimal solution as
\begin{equation} \label{eq::optimal_r} \small
	\mathbf{r}_{B}^{\rm opt} = \frac{\left(\mathbf{I}_{K+M} \otimes \mathbf{H}_t(\mathbf{\Phi})\right)\tilde{\mathbf{w}}}{\tilde{\mathbf{w}}^{\rm H} \mathbf{I}_{K+M} \otimes \mathbf{H}^{\rm H}_t(\mathbf{\Phi})\mathbf{H}_t(\mathbf{\Phi}) \tilde{\mathbf{w}}}.
\end{equation}
In addition, it is noted that $\mathbf{r}_{B}^{\rm opt}e^{j\theta}$ also serves as an optimal solution to \eqref{eq::Subproblem_rBS} for any given phase $\theta$, as the phase of $\mathbf{r}_B^{\rm H}\tilde{\mathbf{y}}_r$ will not affect the result in \eqref{eq::echo_signals_L_filter_vectorized_receivedBeamformer}. As such, after obtaining the optimal receive beamformers, $\mathbf{r}_B^{\rm H}\left(\mathbf{I}_{K+M} \otimes \mathbf{H}_t(\mathbf{\Phi})\right)\tilde{\mathbf{w}}$ in \eqref{eq::ReFormulation_b} is restricted to  non-negative real values, and thus \eqref{eq::ReFormulation_b} can be transformed into 
\begin{equation} \small \label{eq::sensing_constraint_new}
	\Re\{\mathbf{r}_B^{\rm H}\left(\mathbf{I}_{K+M} \otimes \mathbf{H}_t(\mathbf{\Phi})\right)\tilde{\mathbf{w}}\} \geq \sqrt{\frac{\Gamma_k\sigma_r^2\mathbf{r}_{B}^{\rm H}\mathbf{r}_{B}}{L\sigma_t^2}}.
\end{equation}

%\begin{subequations}  \label{eq::Subproblem_1}
%	\begin{align}
%		\max_{\mathbf{r}_{B}} & \frac{\mathbf{r}_{B}^{\rm H}\mathbf{B}^{\rm H}(\mathbf{W},\mathbf{\Phi})\mathbf{B}(\mathbf{W},\mathbf{\Phi})\mathbf{r}_{B}}{\sigma_b^2\mathbf{r}_{B}^{\rm H}\mathbf{r}_{B}} \label{eq::Subproblem_1_objective}\\
%		{\rm s.t.}~ & \left\|\mathbf{r}_{B}^{\rm H}\right\|_2 = 1,
%		\label{eq::Subproblem_1_a}
%	\end{align}
%\end{subequations}
%where $\mathbf{B}^{\rm H}(\mathbf{W},\mathbf{\Phi})$ is given by
%\begin{equation} \small
%	\mathbf{B}^{\rm H}(\mathbf{W},\mathbf{\Phi}) = \mathbf{H}^{\rm H}_{BR}\mathbf{\Phi}^{\rm H}\mathbf{A}_{RE}\mathbf{\Phi}\mathbf{H}_{BR}\mathbf{W}.
%\end{equation}
%It is noted that problem \eqref{eq::Subproblem_1} is a typical Rayleigh quotient [x], whose optimal solution can be easily obtained as the eigenvector corresponding to the largest eigenvalue of the matrix $\frac{\mathbf{H}(\mathbf{W},\mathbf{\Phi})}{\sigma_b^2}$.

\subsubsection{Optimization of $\mathbf{W}$ with given $\Omega$, $\mathbf{r}_B$, $\mathbf{\Phi}$, and $\{\mathbf{u}_k\}_{k=1}^{K}$}
The subproblem of optimizing transmit beamforming $\mathbf{W}$ at the BS can be represented by

\vspace{-1 em}
\begin{subequations}  \label{eq::subproblem_W} \small
	\begin{align}
		\min_{\mathbf{W}}& \sum_{j=1}^{K+M}\!\!\left|\mathbf{g}_0^{\rm H}(\mathbf{\Phi})\mathbf{w}_{j} \!-\! x_{j}\right|^2  \!\!+\!\! \sum_{k=1}^{K}\!\!\!\sum_{j=1}^{K+M}\!\!\left|\mathbf{g}^{\rm H}_k(\mathbf{\Phi},\mathbf{u}_k)\mathbf{w}_j \!-\! z_{k,j}\right|^2 \\
		 {\rm s.t.}~  
%		 &	\Re\{\mathbf{r}_B^{\rm H}\left(\mathbf{I}_{K+M} \otimes \mathbf{H}_t(\mathbf{\Phi})\right)\tilde{\mathbf{w}}\} \geq \sqrt{\frac{\Gamma_k\sigma_r^2\mathbf{r}_{B}^{\rm H}\mathbf{r}_{B}}{L\sigma_t^2}}, \label{eq::subproblem_constraint_1}\\
		&\eqref{eq::sensing_constraint_new},\eqref{eq::originalFormulation_e}, \nonumber
	\end{align}
\end{subequations}
which is convex and can be efficiently addressed using various well-developed toolboxes, such as CVX \cite{boyd2004convex}.

\subsubsection{Optimization of $\mathbf{\Phi}$ with given $\Omega$, $\mathbf{r}_B$, $\mathbf{W}$, and $\{\mathbf{u}_k\}_{k=1}^{K}$}
The subproblem for optimizing $\mathbf{\Phi}$ can be expressed as

\vspace{-1 em}
\begin{subequations}  \label{eq::subproblem_Phi} \small
	\begin{align}
		\min_{\mathbf{\Phi}}& \sum_{j=1}^{K+M}\!\!\left|\mathbf{g}_0^{\rm H}(\mathbf{\Phi})\mathbf{w}_{j} \!-\! x_{j}\right|^2  \!\!+\!\! \sum_{k=1}^{K}\!\!\!\sum_{j=1}^{K+M}\!\!\left|\mathbf{g}^{\rm H}_k(\mathbf{\Phi},\mathbf{u}_k)\mathbf{w}_j \!-\! z_{k,j}\right|^2 \label{eq::subproblem_Phi_objective}\\
		{\rm s.t.}~  
%		&	\Re\{\mathbf{r}_B^{\rm H}\left(\mathbf{I}_{K+M} \otimes \mathbf{H}_t(\mathbf{\Phi})\right)\tilde{\mathbf{w}}\} \geq \sqrt{\frac{\Gamma_k\sigma_r^2\mathbf{r}_{B}^{\rm H}\mathbf{r}_{B}}{L\sigma_t^2}}, \label{eq::subproblem_Phi_constraint}\\
		&\eqref{eq::sensing_constraint_new},\eqref{eq::originalFormulation_g}. \nonumber
	\end{align}
\end{subequations}
Because of the implicit function w.r.t $\mathbf{\Phi}$ in both \eqref{eq::subproblem_Phi_objective} and \eqref{eq::sensing_constraint_new}, along with the non-convex unit-modulus constraint as in \eqref{eq::originalFormulation_g}, problem \eqref{eq::subproblem_Phi} cannot be directly addressed.
To tackle these issues, we first propose to handle \eqref{eq::subproblem_Phi_objective} by constructing a convex surrogate function
%to upper-bound it
via the MM technique \cite{sun2016majorization}. Specifically, for quadratic function $\boldsymbol{\phi}^{\rm H}\mathbf{A}\boldsymbol{\phi}$,  its surrogate function at any given point $\boldsymbol{\phi}^{(t)} = [\phi_1^{(t)}, \cdots, \phi_N^{(t)}]$, denoted by $f(\boldsymbol{\phi}|\boldsymbol{\phi}^{(t)})$, can be given by
\begin{equation} \label{eq::surrogate_phi} \small
	\begin{aligned}
		f(\boldsymbol{\phi}|\boldsymbol{\phi}^{(t)}) =  &\lambda_{\max}\boldsymbol{\phi}^{\rm H}\boldsymbol{\phi} - 2\Re\{\boldsymbol{\phi}^{\rm H}(\lambda_{\max}\mathbf{I}_N-\mathbf{A})\boldsymbol{\phi}^{(t)}\} \\
		& + (\boldsymbol{\phi}^{(t)})^{\rm H}(\lambda_{\max}\mathbf{I}_N-\mathbf{A})\boldsymbol{\phi}^{(t)},
	\end{aligned}
\end{equation}
where $\mathbf{A} \in \mathbb{C}^{N \times N}$ is positive semi-definite and $\lambda_{\max}$ represents its largest eigenvalue. As such, we expand the quadratic terms of the objective function in \eqref{eq::subproblem_Phi} and recall that we have $\mathbf{h}_{\kappa}^{\rm H}\mathbf{\Phi} = \boldsymbol{\phi}^{\rm T}{\rm diag}\{\mathbf{h}_{\kappa}^{\rm H}\}$ for $\forall \kappa \in \{0,1,\cdots,K\}$ since $\mathbf{\Phi}$ is diagonal. Furthermore, based on the fact of  $\boldsymbol{\phi}^{\rm H}\boldsymbol{\phi} = N$ and by neglecting the constant terms w.r.t. $\{\phi_n\}_{n=1}^{N}$, we can turn to solve the  approximate optimization problem as follows

\vspace{-1 em}
\begin{subequations} \small  \label{eq::subproblem_phi_objective_approximation}
	\begin{align}
		\max_{\boldsymbol{\phi}}& \quad \Re\{\boldsymbol{\phi}^{\rm H}\mathbf{q}^{(t)}\}, \\
		{\rm s.t.}&~ \quad   \eqref{eq::sensing_constraint_new},\eqref{eq::originalFormulation_g}, \nonumber
	\end{align}
\end{subequations}
where $\mathbf{q}^{(t)}$ is computed as
\begin{equation} \label{eq::q} \small
	\begin{aligned}
		\mathbf{q}^{(t)} &= \sum_{j=1}^{K+M}\left(\left(\lambda_{\max,1,j}\mathbf{I}_N - \mathbf{\Upsilon}_{0,j}\right)\boldsymbol{\phi}^{(t)} + x_j^{\rm H}\mathbf{F}_0\mathbf{w}_j \right) \\
		& \!\!+ \sum_{k=1}^{K}\sum_{j=1}^{K+M}\left(\left(\lambda_{\max,2,k,j}\mathbf{I}_N - \mathbf{\Upsilon}_{k,j}\right)\boldsymbol{\phi}^{(t)} + \mathbf{F}_k\mathbf{w}_jz_{k,j} \right),
	\end{aligned}
\end{equation}
with $\mathbf{F}_k \!\!=\!\! {\rm diag}\{\mathbf{h}_k^{\rm H}\}\mathbf{H}$, $\mathbf{F}_0 \!\!=\!\! {\rm diag}\{\mathbf{h}_0^{\rm H}\}\mathbf{H}$, $\mathbf{\Upsilon}_{k,j} \!\!=\!\! \mathbf{F}_k\mathbf{w}_j\mathbf{w}_j^{\rm H}\mathbf{F}_k^{\rm H}$, $\mathbf{\Upsilon}_{0,j} \!=\! \mathbf{F}_0\mathbf{w}_j\mathbf{w}_j^{\rm H}\mathbf{F}_0^{\rm H}$, $\lambda_{\max,1,j}$ and $\lambda_{\max,2,k,j}$ are the maximum eigenvalues of $\mathbf{\Upsilon}_{0,j}$ and $\mathbf{\Upsilon}_{k,j}$, respectively.

Second, to tackle \eqref{eq::sensing_constraint_new}, it is necessary to isolate the optimization variable $\mathbf{\Phi}$ from the Kronecker product expression embedded in $\mathbf{H}_t(\mathbf{\Phi})$.
According to the definition of $\mathbf{H}_t(\mathbf{\Phi})$ and utilizing the property of the Kronecker product, i.e., ${\rm vec}\{\mathbf{ABC}\} = (\mathbf{C}^{\rm T} \otimes \mathbf{A}){\rm vec}\{\mathbf{B}\}$, we can have
%transform the term $\left(\mathbf{I}_{K+M} \otimes \mathbf{H}_t(\mathbf{\Phi})\right)\tilde{\mathbf{w}}$ into
\begin{equation} \label{eq::convert_Phi} \small
	\begin{aligned}
		&\left(\mathbf{I}_{K+M} \otimes \mathbf{H}_t(\mathbf{\Phi})\right)\tilde{\mathbf{w}} = \left(\mathbf{I}_{K+M} \otimes (\mathbf{F}_0^{\rm H}\boldsymbol{\phi}\boldsymbol{\phi}^{\rm T}\mathbf{F}_0) \right)\tilde{\mathbf{w}} \\
		&\!=\! {\rm vec}\left\{\mathbf{F}_0^{\rm H}\boldsymbol{\phi}\boldsymbol{\phi}^{\rm T}\mathbf{F}_0\mathbf{W}\mathbf{I}_{K+M}\right\}\!=\! \left(\left(\mathbf{F}_0\mathbf{W}\right)^{\rm T} \otimes \mathbf{F}_0^{\rm H}\right){\rm vec}\{\boldsymbol{\phi}\boldsymbol{\phi}^{\rm T}\}.
	\end{aligned}
\end{equation}
Substituting \eqref{eq::convert_Phi} into \eqref{eq::sensing_constraint_new}, we have
\begin{equation} \label{eq::SNR_radar_constraint_transformation_1} \small
	\Re\{\boldsymbol{\phi}^{\rm T}\tilde{\mathbf{C}}_0\boldsymbol{\phi}\} \geq \sqrt{\frac{\Gamma_k\sigma_r^2\mathbf{r}_{B}^{\rm H}\mathbf{r}_{B}}{L\sigma_t^2}},
\end{equation}
where $\tilde{\mathbf{C}}_0 \in \mathbb{C}^{N \times N}$ represents a reshaped matrix from $\mathbf{c}_0^{\rm T} \triangleq  (\left(\mathbf{F}_0\mathbf{W}\right)^{\rm T} \otimes \mathbf{F}_0^{\rm H})^{\rm T} \mathbf{r}_B^*$ such that $\mathbf{c}_0^{\rm T} = {\rm vec}\{\tilde{\mathbf{C}}_0\}$. However, \eqref{eq::SNR_radar_constraint_transformation_1} remains non-convex because the left-hand side of \eqref{eq::SNR_radar_constraint_transformation_1} is a non-concave function. To tackle this issue, we transform the complex-valued function $\Re\{\boldsymbol{\phi}^{\rm T}\tilde{\mathbf{C}}_0\boldsymbol{\phi}\}$ into the real-valued term $\overline{\boldsymbol{\phi}}^{\rm T}\overline{\mathbf{C}}_0 \overline{\boldsymbol{\phi}}$, which is achieved by defining $\overline{\boldsymbol{\phi}} \triangleq [\Re\{\boldsymbol{\phi}^{\rm T}\} \quad \Im\{\boldsymbol{\phi}^{\rm T}\}]^{\rm T}$ and $\overline{\mathbf{C}}_0 \triangleq \bigg[\begin{array}{cc}
	\Re\{\tilde{\mathbf{C}}_0\} & -\Im\{\tilde{\mathbf{C}}_0\} \\
	-\Im\{\tilde{\mathbf{C}}_0\} & -\Re\{\tilde{\mathbf{C}}_0\}                                 \end{array}\bigg]$.
Then, we apply the MM algorithm \cite{sun2016majorization} to derive a series of tractable surrogate functions for $\overline{\boldsymbol{\phi}}^{\rm T}\overline{\mathbf{C}}_0 \overline{\boldsymbol{\phi}}$. Specifically, given solution $\boldsymbol{\phi}^{(t)}$ in $t$-th iteration, an approximate lower bound is established using the first-order Taylor expansion as follows
\begin{equation} \small \label{eq::first_taylor}
	\begin{aligned}
		\overline{\boldsymbol{\phi}}^{\rm T}\overline{\mathbf{C}}_0 \overline{\boldsymbol{\phi}} & \geq (\overline{\boldsymbol{\phi}}^{(t)})^{\rm T}\overline{\mathbf{C}}_0 \overline{\boldsymbol{\phi}}^{(t)} + (\overline{\boldsymbol{\phi}}^{(t)})^{\rm T}(\overline{\mathbf{C}}_0+\overline{\mathbf{C}}_0^{\rm T})(\overline{\boldsymbol{\phi}}-\overline{\boldsymbol{\phi}}^{(t)}) \\
		& = -(\overline{\boldsymbol{\phi}}^{(t)})^{\rm T}\overline{\mathbf{C}}_0 \overline{\boldsymbol{\phi}}^{(t)} + \Re\{(\overline{\boldsymbol{\phi}}^{(t)})^{\rm T}(\overline{\mathbf{C}}_0+\overline{\mathbf{C}}_0^{\rm T})\boldsymbol{\Pi}\boldsymbol{\phi}\} \\
		& = -\varepsilon_2^{(t)} + \Re\{\tilde{\boldsymbol{\Pi}}^{(t)}\boldsymbol{\phi}\},
	\end{aligned}
\end{equation}
with $\boldsymbol{\Pi} = [\mathbf{I}_N \quad j\mathbf{I}_N]^{\rm T}$, $\varepsilon_2^{(t)} = (\overline{\boldsymbol{\phi}}^{(t)})^{\rm T}\overline{\mathbf{C}}_0 \overline{\boldsymbol{\phi}}^{(t)}$, and $\tilde{\boldsymbol{\Pi}}^{(t)} = (\overline{\boldsymbol{\phi}}^{(t)})^{\rm T}(\overline{\mathbf{C}}_0+\overline{\mathbf{C}}_0^{\rm T})\boldsymbol{\Pi}$. As such, in each iteration, the radar SNR constraint can be rewritten by substituting \eqref{eq::first_taylor} into \eqref{eq::SNR_radar_constraint_transformation_1} as 
\begin{equation} \label{eq::SNR_radar_constraint_convex} \small
	\Re\{\mathbf{d}\boldsymbol{\phi}\} \geq \varepsilon_3,
\end{equation}
with $\mathbf{d}\!=\! \tilde{\boldsymbol{\Pi}}^{(t)}$ and $\varepsilon_3 \!=\! \sqrt{\frac{\Gamma_k\sigma_r^2\mathbf{r}_{B}^{\rm H}\mathbf{r}_{B}}{L\sigma_t^2}} + \varepsilon_2^{(t)}$. Thus, \eqref{eq::SNR_radar_constraint_convex} now becomes convex.

Finally, problem \eqref{eq::subproblem_Phi} only involves the non-convex unit-modulus constraint as in \eqref{eq::originalFormulation_g}. In order to tackle this issue, we relax constraint \eqref{eq::originalFormulation_g} as \cite{hua2021joint}
\begin{equation} \label{eq::relaxed_unit-modulus_constraint} \small
	\left|\phi_n\right| \leq 1, \quad 1 \leq n \leq N.
\end{equation}
Based on the previous analysis, problem \eqref{eq::subproblem_Phi} can be rewritten by dropping the constant terms as

\vspace{-1 em}
\begin{subequations}  \label{eq::subproblem_Phi_transformation} \small
	\begin{align}
		\max_{\boldsymbol{\phi}}& \quad  \Re\{\boldsymbol{\phi}^{\rm H}\mathbf{q}^{(t)}\} \\
		{\rm s.t.}& \quad \eqref{eq::SNR_radar_constraint_convex}, \eqref{eq::relaxed_unit-modulus_constraint}, \nonumber
%		{\rm s.t.}& \quad 	\Re\{\mathbf{d}\boldsymbol{\phi}\} \geq \varepsilon_3, \label{eq::subproblem_Phi_transformation_constraint}\\
%		&\quad \left|\phi_n\right| \leq 1, \quad 1 \leq n \leq N,
	\end{align}
\end{subequations}
which is convex and can be efficiently addressed using various well-developed algorithms or toolboxes \cite{boyd2004convex}. Note that according to \cite{zhu2018joint}, the optimal solution for problem \eqref{eq::subproblem_Phi_transformation} always satisfies the unit-modulus constraint \eqref{eq::originalFormulation_g}.

\subsubsection{Optimization of $\{\mathbf{u}_k\}_{k=1}^{K}$ with given $\Omega$, $\mathbf{r}_B$, $\mathbf{W}$, and $\mathbf{\Phi}$}
The subproblem for optimizing the positions of MAs $\{\mathbf{u}_k\}_{k=1}^{K}$ for $K$ users is given by

\vspace{-1 em}
\begin{subequations}  \label{eq::subproblem_uk} \small
	\begin{align}
		 \min_{\{\mathbf{u}_k\}}& \quad \sum_{k=1}^{K}\sum_{j=1}^{K+M}\left|\mathbf{g}^{\rm H}_k(\mathbf{\Phi},\mathbf{u}_k)\mathbf{w}_j - z_{k,j}\right|^2 \label{eq::subproblem_uk_a}\\
		{\rm s.t.}&~ \quad  \eqref{eq::originalFormulation_h}. \nonumber
	\end{align}
\end{subequations}
The main difficulties for tackling problem \eqref{eq::subproblem_uk} lie in the intractability of the objective function and the tight coupling of $\{\mathbf{u}_k\}_{k=1}^{K}$. Thus, we propose to divide $\{\mathbf{u}_k\}_{k=1}^{K}$ into $K$ blocks and optimize the position of MA at $k$-th user, $\mathbf{u}_k$, while keeping all other variables fixed, i.e., $\{\mathbf{u}_{k'}, k' \neq k\}_{k'=1}^{K}$. Moreover, since the variation of $\mathbf{u}_k$ only affects the receive FRM $\mathbf{f}_k(\mathbf{u}_k)$ shown in \eqref{eq::channel_RIS_to_kappa}, we expand \eqref{eq::subproblem_uk_a} as follows
\begin{equation} \small
	 \Xi(\mathbf{u}_k) = \mathbf{f}_k^{\rm H}(\mathbf{u}_k) \mathbf{Q}_k \mathbf{f}_k(\mathbf{u}_k) + 2\Re\left\{\mathbf{f}_k^{\rm H}(\mathbf{u}_k)\mathbf{p}_k\right\} + \varepsilon_4,
\end{equation}
with $\mathbf{Q}_k \triangleq \mathbf{\Sigma}_k\mathbf{G}_k\mathbf{\Phi}\mathbf{H}\overline{\mathbf{W}}\mathbf{H}^{\rm H}\mathbf{\Phi}^{\rm H}\mathbf{G}_k^{\rm H}\mathbf{\Sigma}_k^{\rm H}$, $\mathbf{p}_k \triangleq -z_{k,j}^*\mathbf{\Sigma}_k\mathbf{G}_k\mathbf{\Phi}\mathbf{H}\overline{\mathbf{w}}$, $\overline{\mathbf{W}} \triangleq \sum_{j=1,j \neq k}^{K+M}\mathbf{w}_j\mathbf{w}_j^{\rm H}$, $\overline{\mathbf{w}} \triangleq \sum_{j=1,j \neq k}^{K+M}\mathbf{w}_j$, and $\varepsilon_4 = |z_{k,j}|^2$ is constant w.r.t. $\mathbf{u}_k$. Then, by neglecting  the constant terms from $\Xi(\mathbf{u}_k)$, the subproblem of optimizing MA position for user $k$ can be expressed as

\vspace{-1 em}
\begin{subequations}  \label{eq::subproblem_uk_reformulation} \small
	\begin{align}
		\min_{\mathbf{u}_k}& \quad \mathbf{f}_k^{\rm H}(\mathbf{u}_k) \mathbf{Q}_k \mathbf{f}_k(\mathbf{u}_k) + 2\Re\left\{\mathbf{f}_k^{\rm H}(\mathbf{u}_k)\mathbf{p}_k\right\} \label{eq::subproblem_uk_reformulation_objective}\\
		{\rm s.t.}&~ \quad  \mathbf{u}_k \in \mathcal{C}_k, \label{eq::subproblem_uk_reformulation_b}
	\end{align}
\end{subequations}
which is non-convex due to \eqref{eq::subproblem_uk_reformulation_objective}. Next, we tackle this issue by employing the MM algorithm  \cite{sun2016majorization}. To construct the surrogate function for \eqref{eq::subproblem_uk_reformulation_objective}, we introduce the  lemma as follows.
\begin{lemma} \label{lemma::2}
	For given $\mathbf{u}_k^{(t)}$ in the $t$-th iteration, we have
	\begin{equation} \label{eq::lemma_2} \small
		\begin{aligned}
			 \mathbf{f}_k^{\rm H}(\mathbf{u}_k) \mathbf{Q}_k \mathbf{f}_k(\mathbf{u}_k) &\leq \mathbf{f}_k^{\rm H}(\mathbf{u}_k) \boldsymbol{\Lambda}_k \mathbf{f}_k(\mathbf{u}_k) \\
			&\!\!\!\!\!\!\!\!\!\!\!\! - 2\Re\{\mathbf{f}_k^{\rm H}(\mathbf{u}_k) (\boldsymbol{\Lambda}_k - \mathbf{Q}_k) \mathbf{f}_k(\mathbf{u}_k^{(t)})\} \\
			&\!\!\!\!\!\!\!\!\!\!\!\! + \mathbf{f}_k^{\rm H}(\mathbf{u}_k^{(t)}) (\boldsymbol{\Lambda}_k - \mathbf{Q}_k) \mathbf{f}_k(\mathbf{u}_k^{(t)}) \triangleq \omega_k(\mathbf{u}_k|\mathbf{u}_k^{(t)}),
		\end{aligned}
	\end{equation}
	where $\boldsymbol{\Lambda}_k \triangleq \lambda_{\max,3,k}\mathbf{I}_{L_k^r}$, and $\lambda_{\max,3,k}$ denotes the maximum eigenvalue of $\mathbf{Q}_k$.
\end{lemma}
\begin{proof}
	Refer to Appendix \ref{app0}.
\end{proof}

In light of Lemma \ref{lemma::2}, the surrogate function for \eqref{eq::subproblem_uk_reformulation_objective} can be readily constructed as
\begin{equation} \label{eq::surrogate_function_uk} \small
	\psi_k(\mathbf{u}_k|\mathbf{u}_k^{(t)}) = \omega_k(\mathbf{u}_k|\mathbf{u}_k^{(t)}) + 2\Re\left\{\mathbf{f}_k^{\rm H}(\mathbf{u}_k)\mathbf{p}_k\right\}.
\end{equation}
In addition, according to $\mathbf{f}_k^{\rm H}(\mathbf{u}_k) \mathbf{f}_k(\mathbf{u}_k) = L^k_r$, we have $\mathbf{f}_k^{\rm H}(\mathbf{u}_k) \boldsymbol{\Lambda}_k \mathbf{f}_k(\mathbf{u}_k) = \lambda_{\max,3,k}L^k_r$, which is a constant. Then, $\psi_k(\mathbf{u}_k|\mathbf{u}_k^{(t)})$ can be recast as
\begin{equation} \label{eq::surrogate_function_uk_2} \small
	\psi_k(\mathbf{u}_k|\mathbf{u}_k^{(t)}) = 2\Re\left\{\mathbf{f}_k^{\rm H}(\mathbf{u}_k)\boldsymbol{\varsigma}_k^{(t)}\right\} + \varepsilon_5,
\end{equation}
where $\boldsymbol{\varsigma}_k^{(t)} = \mathbf{p}_k - (\boldsymbol{\Lambda}_k - \mathbf{Q}_k) \mathbf{f}_k(\mathbf{u}_k^{(t)})$, and $\varepsilon_5 = L_k^r + \mathbf{f}_k^{\rm H}(\mathbf{u}_k^{(t)}) (\boldsymbol{\Lambda}_k - \mathbf{Q}_k) \mathbf{f}_k(\mathbf{u}_k^{(t)})$ is constant w.r.t. $\mathbf{u}_k$. However,  $\psi_k(\mathbf{u}_k|\mathbf{u}_k^{(t)})$ is still non-concave w.r.t. $\mathbf{u}_k$.

Next, the second-order Taylor expansion of $\psi_k(\mathbf{u}_k|\mathbf{u}_k^{(t)}) \triangleq \hat{\psi}_k(\mathbf{u}_k)$ is employed to construct the surrogate function. Specifically, we introduce $\delta_k > 0$ satisfying $\delta_k\mathbf{I}_2 \succeq \nabla^2\hat{\psi}_k(\mathbf{u}_k)$ \cite{magnus1995matrix}, and thus we obtain
\begin{equation} \label{eq::surrogate_function_uk_twice} \small
	\begin{aligned}
		\hat{\psi}_k(\mathbf{u}_k) &\leq \hat{\psi}_k(\mathbf{u}_k^{(t)}) + \nabla\hat{\psi}_k(\mathbf{u}_k^{(t)})^{\rm T}(\mathbf{u}_k-\mathbf{u}_k^{(t)}) \\
		& \quad + \frac{\delta_k}{2}(\mathbf{u}_k-\mathbf{u}_k^{(t)})^{\rm T}(\mathbf{u}_k-\mathbf{u}_k^{(t)}) \\
		& = \frac{\delta_k}{2}\mathbf{u}_k^{\rm T}\mathbf{u}_k + (\nabla\hat{\psi}_k(\mathbf{u}_k^{(t)}) - \delta_k\mathbf{u}_k^{(t)})^{\rm T}\mathbf{u}_k \\
		& \quad + \hat{\psi}_k(\mathbf{u}_k^{(t)}) + \frac{\delta_k}{2}(\mathbf{u}_k^{(t)})^{\rm T}\mathbf{u}_k^{(t)} - \nabla\hat{\psi}_k(\mathbf{u}_k^{(t)})^{\rm T}\mathbf{u}_k^{(t)},
	\end{aligned}
\end{equation}
where $\nabla\hat{\psi}_k(\mathbf{u}_k)$ and  $\nabla^2\hat{\psi}_k(\mathbf{u}_k)$ are provided in Appendix \ref{app1}, and Appendix \ref{app2} provides the selection method of $\delta_k$.

Finally, by utilizing the surrogate function as defined in \eqref{eq::surrogate_function_uk_twice} and dropping the constant terms w.r.t. $\mathbf{u}_k$, we can transformed problem \eqref{eq::subproblem_uk_reformulation} into

\vspace{-1 em}
\begin{subequations} \small \label{eq::subproblem_uk_reformulation_with_convex}
	\begin{align}
		\min_{\mathbf{u}_k}& \quad \frac{\delta_k}{2}\mathbf{u}_k^{\rm T}\mathbf{u}_k + (\nabla\hat{\psi}_k(\mathbf{u}_k^{(t)}) - \delta_k\mathbf{u}_k^{(t)})^{\rm T}\mathbf{u}_k \label{eq::subproblem_uk_reformulation_with_convex_objecitve}\\
		{\rm s.t.}&~ \quad  \eqref{eq::subproblem_uk_reformulation_b}. \nonumber
	\end{align}
\end{subequations}
Since \eqref{eq::subproblem_uk_reformulation_with_convex_objecitve} is a quadratic convex function w.r.t. $\mathbf{u}_k$ and the antenna moving region $\mathcal{C}_k$ is linear, \eqref{eq::subproblem_uk_reformulation_with_convex} is a quadratic programming problem. This indicates that we can obtain a local optimal solution using various well-developed algorithms or toolboxes \cite{boyd2004convex}. Based on the above discussions, the detailed steps for solving \eqref{eq::subproblem_uk_reformulation} are summarized in \textbf{Algorithm}~\ref{alg1}. 
%When the algorithm converges, we can obtain the optimal MA position for user $k$.

\subsection{Outer Layer Optimization}
When the inner-layer optimization converges, the penalty factor in the $T$-th iteration of the outer-layer can be updated as
\begin{equation} \label{eq::penalty_factor_update} \small
	\rho^T = \eta\rho^{T-1}, 0 < \eta < 1,
\end{equation}
where $\eta$ is a constant. Besides, it is noted that increasing $\eta$ can improve performance, but it requires a greater number of iterations.

\subsection{Overall Solution}\label{SecC}
We summarize the detailed procedures of the proposed overall solution in \textbf{Algorithm}~\ref{alg2}. Specifically, in steps 4-8, we first optimize the auxiliary variable set $\Omega$ with other variables fixed. Then, in step 9, we obtain the receive filter $\mathbf{r}_B$ in closed form. The transmit beamformers $\mathbf{W}$ are then determined by addressing the convex problem \eqref{eq::subproblem_W} in step 10. Next, we optimize the RIS coefficient matrix sequentially by solving problem \eqref{eq::subproblem_Phi_transformation} based on the MM algorithm in step 11. In steps 12-14, the positions of MAs for $K$ users are optimized by successively solving the problem \eqref{eq::subproblem_uk_reformulation_with_convex}. The inner-layer iterations alternately address the aforementioned five subproblems until the objective value in \eqref{eq::ReFormulation_2} increases by less than $\epsilon_{\rm in}$ or the number of inner-layer iterations exceeds $T_1^{\rm max}$. Meanwhile,  the outer-layer iterations adjust the penalty factor and will terminate if the iteration count exceeds $T_2^{\rm max}$ or the following condition is satisfied
\begin{equation} \label{eq::termination_indicator} \small
	\max_{\forall i,k} \left\{\left|\mathbf{g}_0^{\rm H}(\mathbf{\Phi})\mathbf{w}_{j} - x_{j}\right|^2, \left|\mathbf{g}^{\rm H}_k(\mathbf{\Phi},\mathbf{u}_k)\mathbf{w}_j - z_{k,j}\right|^2 \right\} \leq \epsilon_{\rm out},
\end{equation}
indicating that  \eqref{eq::ReFormulation_e} is met with equality for a given accuracy.

\begin{algorithm}[t] \small
	\caption{MM Algorithm for Solving \eqref{eq::subproblem_uk_reformulation}}
	\label{alg1}
	\begin{algorithmic}[1]
		\STATE \textbf{Initialize} convergence threshold $\epsilon_1$ and iteration index $t = 1$.
		\REPEAT
		\STATE Calculate $\boldsymbol{\Lambda}_k$ according to Lemma \ref{lemma::2}.
		\STATE Calculate $\boldsymbol{\varsigma}_k^{(t)}$.
		\STATE Calculate $\nabla\hat{\psi}_k(\mathbf{u}_k^{(t)})$, $\nabla^2\hat{\psi}_k(\mathbf{u}_k^{(t)})$, and $\delta_k$ via \eqref{eq::first_order}, \eqref{eq::second_order}, and \eqref{eq::delta}, respectively.
		\STATE Obtain $\mathbf{u}_k^{(t+1)}$ by solving problem  \eqref{eq::subproblem_uk_reformulation_with_convex}.
		\STATE Set $t$ $\leftarrow$ $t+1$.
		\UNTIL Increase of the value in \eqref{eq::subproblem_uk_reformulation_objective} is below $\epsilon_1$.
		\RETURN $\mathbf{u}_k$.
	\end{algorithmic}
\end{algorithm}

\begin{algorithm}[t] \small
	\caption{Penalty-Based Algorithm for Solving \eqref{eq::originalFormulation}}
	\label{alg2}
	\begin{algorithmic}[1]
		\STATE  \textbf{Initialize}  $\boldsymbol{\lambda}^{(T_1)}$, $\boldsymbol{\iota}^{(T_1)}$, $\rho^{(T_2)}$, $\{{z_{k,j}^{(T_1)},x_j^{(T_1)}}\}$, $\mathbf{r}_B^{(T_1)}$, $\mathbf{W}^{(T_1)}$, $\mathbf{\Phi}^{(T_1)}$, $\{\mathbf{u}_k^{(T_1)}\}_{k=1}^{K}$,  $\epsilon_{\rm in}$, $\epsilon_{\rm out}$,  inner-layer iteration index $T_1 = 0$, outer-layer iteration index $T_2 = 0$, $T_1^{\rm max}$, and $T_2^{\rm max}$.
		\REPEAT
			\REPEAT
			\STATE Update $\boldsymbol{\lambda}^{(T_1+1)}$ and $\boldsymbol{\iota}^{(T_1+1)}$ via \eqref{eq::optimal_lambda} and \eqref{eq::optimal_c}.
			\FOR{$k=1$ to $K$}  
			\STATE Update  $\{z_{k,j}^{(T_1+1)}, \forall j\}$  by  solving  problem \eqref{eq::subproblem_Omega_z}.
			\ENDFOR
			\STATE Update $\{x_j^{(T_1+1)}, \forall j\}$  by  solving problem \eqref{eq::subproblem_Omega_x}.
			\STATE Update receive filter $\mathbf{r}_B^{(T_1+1)}$ via \eqref{eq::optimal_r}.
			\STATE Update $\mathbf{W}^{(T_1+1)}$ by addressing problem \eqref{eq::subproblem_W}.
			\STATE Update $\mathbf{\Phi}^{(T_1+1)}$ by addressing problem \eqref{eq::subproblem_Phi_transformation}.
			\FOR{$k=1$ to $K$}  
			\STATE Update $\mathbf{u}_k^{(T_1+1)}$ via Algorithm \ref{alg1}.
			\ENDFOR 
			\STATE Set $T_1 \leftarrow T_1 + 1$.
			\UNTIL the fractional increase of the objective value in \eqref{eq::ReFormulation_2} is below the threshold $\epsilon_{\rm in}$ or $T_1 > T_1^{\rm max}$.
		\STATE Update penalty factor $\rho^{(T_2+1)}$ according to \eqref{eq::penalty_factor_update}.
		\STATE Set $T_2 \leftarrow T_2 + 1$ and $T_1 \leftarrow 0$.
		\UNTIL the condition in \eqref{eq::termination_indicator} is satisfied or $T_2 > T_2^{\rm max}$.
	\end{algorithmic}
\end{algorithm}

\subsubsection{Convergence Analysis}
In \textbf{Algorithm}~\ref{alg2}, each subproblem in the inner layer is solved locally and/or optimally, ensuring \eqref{eq::ReFormulation_2_objective} is non-decreasing over iterations. Moreover, since  \eqref{eq::ReFormulation_2_objective} is upper-bounded by the finite transmit power budget, following \cite{hua2023secure}, the solution derived by \textbf{Algorithm}~\ref{alg2} is ensured to converge.
% to a stationary point.

\subsubsection{Computational Complexity}
Note that the interior point method is employed to address the subproblems in \eqref{eq::subproblem_W}, \eqref{eq::subproblem_Phi_transformation}, and \eqref{eq::subproblem_uk_reformulation_with_convex}.
Then, we analyze the complexity of \textbf{Algorithm}~\ref{alg2} as follows. Specifically, in steps 6 and 8, bisection methods are employed, whose computation complexities are obtained as $\mathcal{O}\left(K\log_2(\frac{1}{\epsilon})M^3\right)$and $\mathcal{O}\left(\log_2(\frac{1}{\epsilon})M^3\right)$, respectively, where $\epsilon$ represents the iteration accuracy.
%In step x, the closed-form expression for the receive beamformer vector $\mathbf{r}_B$ is obtained, whose complexity is $\mathcal{O}(1)$.
In step 10, the complexity for obtaining 
$\mathbf{W}$ is  $\mathcal{O}\left(M^{3.5}(M+K)^{3.5}\right)$.
In step 11, the complexity of optimizing  $\mathbf{\Phi}$  is $\mathcal{O}\left(N^{3.5}\right)$.
In step 13, the complexity for obtaining the maximum eigenvalues of $\mathbf{Q}_k$ is given by $\mathcal{O}\left((L^r_k)^3\right)$. The complexities for calculating $\nabla\hat{\psi}_k(\mathbf{u}_k)$, $\nabla^2\hat{\psi}_k(\mathbf{u}_k)$, and $\delta_k$ are given by $\mathcal{O}\left(L^r_k\right)$, $\mathcal{O}\left(L^r_k\right)$, and $\mathcal{O}(1)$, respectively. Besides, updating $\mathbf{u}_k^{(t+1)}$ via solving problem \eqref{eq::subproblem_uk_reformulation_with_convex}, incurs a complexity of $\mathcal{O}\left(K(2^{3.5})\right)$, where 2 represents the number of variables. As such, the total complexity for obtaining $\{\mathbf{u}_k\}_{k=1}^{K}$ is given by $\mathcal{O}\left((L^r_k)^3 + I_{\max}K(2^{3.5})\right)$, where $I_{\max}$ denotes the maximum iterations to solve problem \eqref{eq::subproblem_uk_reformulation_with_convex}.
Hence,  we can readily obtain the total complexity for \textbf{Algorithm}~\ref{alg2}, denoted as $\mathcal{O}\Big(T_1^{\max}(T_2^{\max}(K\log_2(\frac{1}{\epsilon})M^3 +  \log_2(\frac{1}{\epsilon})M^3 + M^{3.5}(M+K)^{3.5} + N^{3.5} + (L^r_k)^3 + I_{\max}K2^{3.5}))\Big)$.
%, where $t_{\max}$ and $T_{\max}$ denote the number of iterations required for convergence in the inner layer and outer layer, respectively.
%In practice, the RIS always contains significantly more elements compared to the number of users and BS antennas, i.e., $N \gg \max\{M, K\}$. As such, the computational complexity of the overall algorithm can be well-approximated by $\mathcal{O}(T_1^{\max}T_2^{\max}N^{3.5})$.
In practice, since we always have $N \gg \max\{M, K\}$, the complexity of the overall algorithm can be well-approximated by $\mathcal{O}(T_1^{\max}T_2^{\max}N^{3.5})$.

\section{Simulation Results}
This section presents simulation results that validate the performance of securing transmission for MAs-aided RIS-ISAC systems. 
%\subsection{Simulation Setup}
%Each user is equipped with a single MA, 
The eavesdropping target has a single FPA-based UPA, and the ISAC BS and RIS are equipped with $M=6$ and $N=64$ FPA-based UPAs, respectively. The locations of the RIS and ISAC BS are set to $(0, 20, 3)$ meters (m) and $(0, 0, 3)$ m, respectively. The users are distributed randomly and uniformly around a circle with a radius of $4$ m, centered at coordinates $(5, 20, 0)$ m. The elevation and azimuth AoAs (AoDs) for users/target (RIS) are considered to be independent identically distributed (i.i.d.) random variables that follow a uniform distribution, i.e., $\vartheta_{\kappa,i}^r$, $\varphi_{\kappa,i}^r$, $\vartheta_{\kappa,j}^t$, $\varphi_{\kappa,j}^t$  $\sim \mathcal{U}[0,\pi]$, $1 \leq i \leq L_{\kappa}^r$, $1 \leq j \leq L_{\kappa}^t$. Similarly, we assume $\vartheta_{s,i}^r$, $\varphi_{s,i}^r$, $\vartheta_{s,j}^t$, $\varphi_{s,j}^t$ $\sim \mathcal{U}[0,\pi]$, $1 \leq i \leq L_{bs}^r$, $1 \leq j \leq L_{bs}^t$.
%$\vartheta_{\kappa,i}^r \sim \mathcal{U}(0,\pi)$, $\varphi_{\kappa,i}^r \sim \mathcal{U}(0,\pi)$, $1 \leq i \leq L_{\kappa}^r$, as well as $\vartheta_{\kappa,j}^t \sim \mathcal{U}(0,\pi)$, $\varphi_{\kappa,j}^t \sim \mathcal{U}(0,\pi)$, $1 \leq j \leq L_{\kappa}^t$, respectively. 
%Similarly, we have $\vartheta_{s,i}^r \sim \mathcal{U}(0,\pi)$, $\varphi_{s,i}^r \sim \mathcal{U}(0,\pi)$, $1 \leq i \leq L_{bs}^r$ and $\vartheta_{s,j}^t \sim \mathcal{U}(0,\pi)$, $\varphi_{s,j}^t \sim \mathcal{U}(0,\pi)$, $1 \leq j \leq L_{bs}^t$.
%Besides, the elevation and azimuth AoAs of the RIS, as well as the elevation and azimuth AoDs of the BS, are also considered to be i.i.d. variables that follow a uniform distribution, i.e., $\vartheta_{s,i}^r \sim \mathcal{U}[0,\pi]$, $\varphi_{s,i}^r \sim \mathcal{U}[0,\pi]$, $1 \leq i \leq L_{bs}^r$ and $\vartheta_{s,j}^t \sim \mathcal{U}[0,\pi]$, $\varphi_{s,j}^t \sim \mathcal{U}[0,\pi]$, $1 \leq j \leq L_{bs}^t$, respectively. 
Moreover, all the channels are described by the geometric channel model \cite{zhu2023modeling}, assuming an equal number of transmit and receive paths, i.e., $L_0^r=L_0^t=1$ and $L_k^r=L_k^t=L_{bs}^r=L_{bs}^t=L_p$, where $L_p=6$. As such, the PRM $\Sigma_{\kappa} \in \mathbb{C}^{L_{\kappa}^r \times L_{\kappa}^t}$ and $\Sigma_{bs} \in \mathbb{C}^{L_{bs}^r \times L_{bs}^t}$ are denoted as diagonal matrices with each element $\sigma_{ll} \sim \mathcal{CN}(0, \frac{g_0d^{-\alpha}}{L_p})$, $1 \leq l \leq L_p$, where $g_0 = -40$ dB and $\alpha = 2.8$, respectively. For simplicity, equal noise variances are assumed across all communication and radar equipment, i.e., $\sigma^2_k = \sigma^2_r =\sigma^2_e = \sigma^2, 1 \leq k \leq K$. The receive region for the MA is modeled as a 2D square, i.e., $\mathcal{C}_k = \mathcal{C} = [-\frac{A}{2},\frac{A}{2}] \times [-\frac{A}{2},\frac{A}{2}]$ for $\forall k$, where $A = 4\lambda$ and $\lambda = 0.01$ m denotes the wavelength. For all the users, the maximum tolerable secrecy leakage SINR and the minimum communication SINR are considered to be identical, i.e., $\Gamma_{e,k}=\Gamma_e = 0$ dB, and $\Gamma_k = \Gamma=10$ dB, $1 \leq k \leq K$. 
Unless otherwise stated, the following simulation parameters are adopted. Specifically, $K=3$, $P_B = 32$ dBm, $\Gamma_r=0$ dB, $\sigma = -90$ dBm, $\sigma_t = 1$, $L=1024$, $\rho = 0.1$, $\eta = 0.85$, $\epsilon_1=\epsilon_{\rm in} = 10^{-7}$, $\epsilon_{\rm out} = 10^{-5}$, $T_1^{\rm max} = 400$, and $T_2^{\rm max} = 120$. We provide simulation results obtained by averaging over 500 independent user distributions and channel realizations.

\subsection{Convergence Performance}
\begin{figure}[t]
	\begin{center}
		\includegraphics[width= 2.8 in]{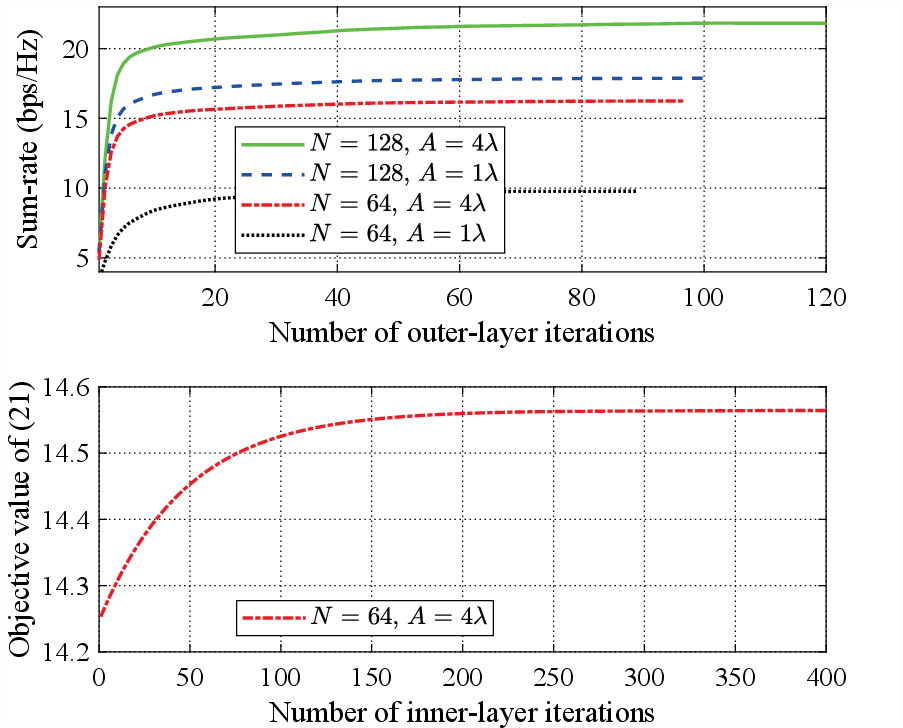}
		\caption{The convergence behavior of outer and inner layer in Algorithm~\ref{alg2}.} \label{fig::convergence_outer}
	\end{center}
	\vspace{-1.5em}
\end{figure}
In Fig. \ref{fig::convergence_outer}, we illustrate the convergence performance of the outer layer and inner layer in \textbf{Algorithm}~\ref{alg2}. It is evident that, regardless of the size of the receive region at users or the number of reflection elements at the RIS, the sum-rate consistently increases and stabilizes after approximately 120 iterations. Specifically, with $N = 128$ and $A = 4\lambda$, the sum-rate improves from $5.44$ bps/Hz to $21.83$ bps/Hz, highlighting the proposed solution's effectiveness in enhancing the PLS of the MAs-aided system under consideration. Additionally, the objective function values for \eqref{eq::ReFormulation_2} converge after 240 iterations, aligning with earlier discussions.

\subsection{Benchmark Schemes and Performance Comparison}
To comprehensively verify the superiority of the proposed scheme (denoted as ``\textbf{Proposed}''), we compared it against several baseline schemes as outlined below.
\begin{itemize}
	\item \textbf{FPA:} Each user's antenna is positioned at the origin within their respective local coordinate systems.
	\item \textbf{Random position antenna (RPA):} The antenna of each user is randomly distributed in its receive region $\mathcal{C}_k, \forall k$.
	\item \textbf{Separate:} The receive filter $\mathbf{r}_B$, communication beamformer $\mathbf{W}_c$, radar beamformer $\mathbf{W}_r$, reflection coefficient matrix $\mathbf{\Phi}$, and MA positions of users $\{\mathbf{u}_k\}_{k=1}^K$ are separately optimized. Specifically, $\mathbf{W}_c$ is optimized by maximizing the sum-rate within a power budget constraint, $\mathbf{W}_r$ is optimized by minimizing the power while ensuring a predefined radar SNR threshold is met, and the optimization of $\mathbf{\Phi}$ involves tackling the sum channel gain maximization problem. Then, $\{\mathbf{u}_k\}_{k=1}^K$ is optimized by maximizing channel power \cite{zhu2023movable}, where each user's MA is positioned to maximize its channel, i.e., $\mathbf{u}_k^* = \arg\max_{\mathbf{u}_k \in \mathcal{C}_k} \|\mathbf{h}_k(\mathbf{u}_k)\|_2^2, 1 \leq k \leq K$. Finally, $\mathbf{r}_B$ is obtained by \eqref{eq::optimal_r}.
	\item \textbf{Communication signal only (Comm only):} The proposed algorithm optimizes multi-user communication by neglecting radar sensing constraints, thereby providing an upper bound on the sum-rate performance for communication users.
	\item \textbf{Random phase:} The RIS phase shifts are generated randomly, following a uniform distribution within the range $[0, 2\pi)$.
\end{itemize}

% 1*3 构图
\begin{figure*}[ht]
	\begin{center}
		\begin{minipage}[b]{0.3\linewidth}
			\centering
			\includegraphics[width= 2.45 in]{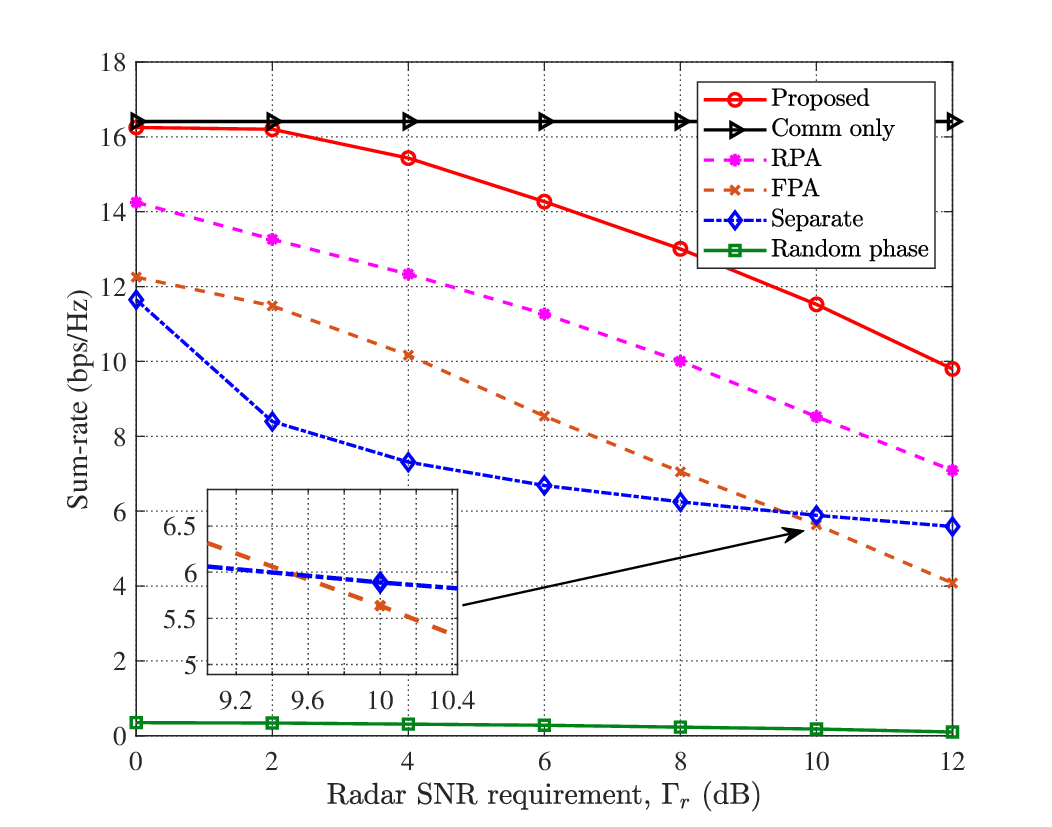}
			\caption{Sum-rate versus the radar SNR threshold $\Gamma_r$.} \label{fig::SR_vs_Gammar}
		\end{minipage}
		\quad 
		\begin{minipage}[b]{0.3\linewidth}
			\centering
			\includegraphics[width= 2.4 in]{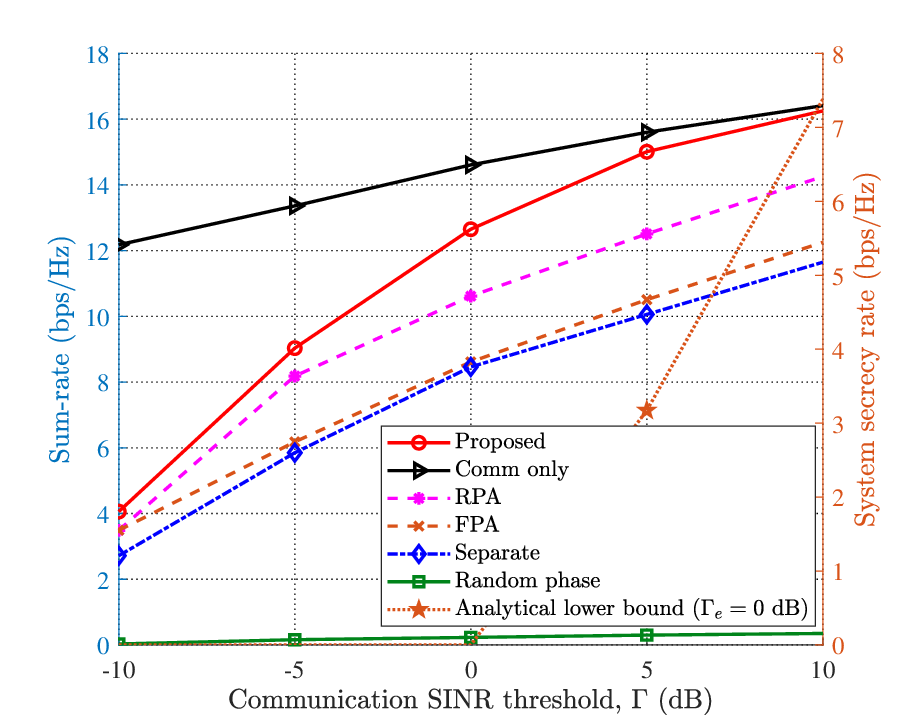}
			\caption{Sum-rate versus the communication SINR threshold $\Gamma$.} \label{fig::SR_vs_Gamma}
		\end{minipage}
		\quad 
		\begin{minipage}[b]{0.3\linewidth}
			\centering
			\includegraphics[width= 2.25 in]{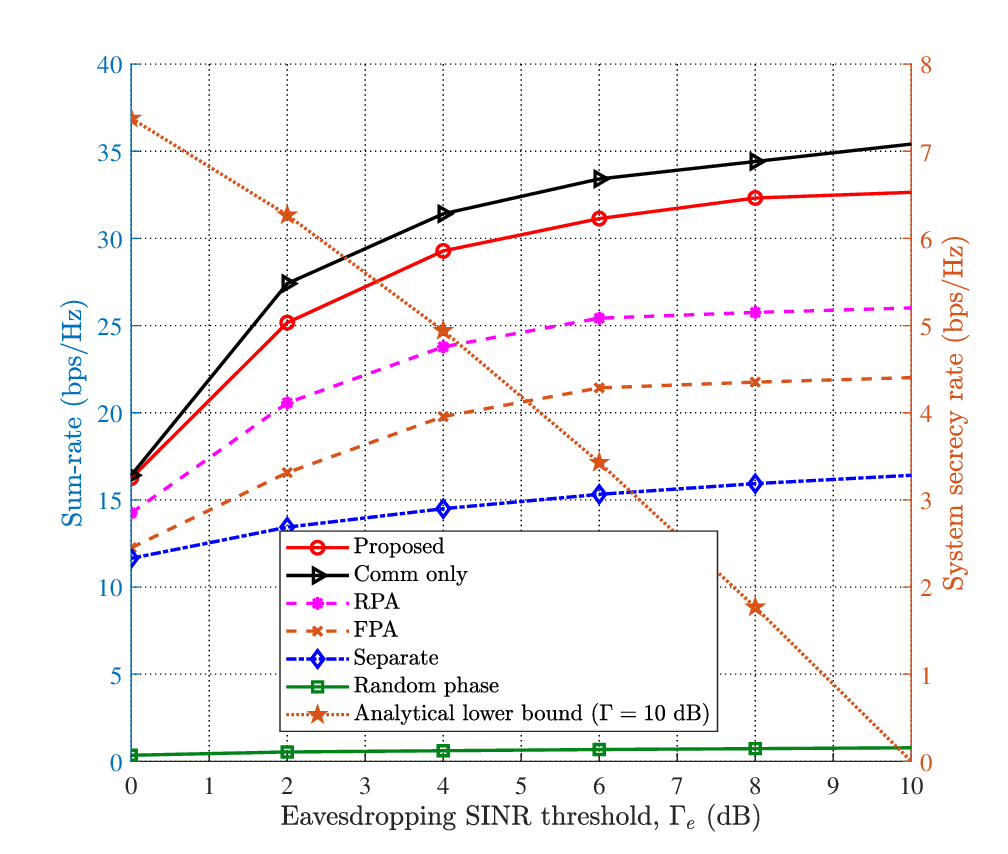}
			\caption{Sum-rate versus the eavesdropping SINR threshold $\Gamma_e$.} \label{fig::SR_vs_Gammae}
		\end{minipage}
	\end{center}
	\vspace{-1.5em}
\end{figure*}

% requirements -> threshold
%\begin{figure}[t]
%	\begin{center}
%		\includegraphics[width= 2.45 in]{figures//SR_vs_Gammar.eps}
%		\caption{Sum-rate versus the radar SNR threshold $\Gamma_r$.} \label{fig::SR_vs_Gammar}
%	\end{center}
%	\vspace{-1.5em}
%\end{figure}
Fig.~\ref{fig::SR_vs_Gammar} compares the sum-rate of the baseline schemes versus radar sensing requirements. It is depicted that the sum-rate for all approaches decreases monotonically with the increase in $\Gamma_r$, except for the ``Comm only'' approach. This is due to the fact that a more stringent target detection constraint necessitates greater power for sensing beamformers, thereby reducing the available power to optimize $\mathbf{W}_c$ for maximizing the sum-rate. Moreover, it is illustrated that the proposed approach surpasses the ``Separate'' approach, implying the sum-rate gain can be achieved through joint optimization approaches. Besides, the ``Random phase'' scheme experiences significant performance deterioration, indicating the performance gain provided by the RIS through channel reconstruction. Finally, the proposed MA scheme outperforms both the  ``RPA'' and ``FPA'' schemes, highlighting the benefits of optimizing MA positions. Consequently, in the considered RIS-ISAC system, implementing MA proves to be highly advantageous for enhancing PLS.

%\begin{figure}[t]
%	\begin{center}
%		\includegraphics[width= 2.4 in]{figures//SR_vs_Gamma_v2.eps}
%		\caption{Sum-rate versus the communication SINR threshold $\Gamma$.} \label{fig::SR_vs_Gamma}
%	\end{center}
%	\vspace{-1.5em}
%\end{figure}
%\begin{figure}[t]
%	\begin{center}
%		\includegraphics[width= 2.4 in]{figures//SR_vs_Gammae_v2.eps}
%		\caption{Sum-rate versus the eavesdropping SINR threshold $\Gamma_e$.} \label{fig::SR_vs_Gammae}
%	\end{center}
%	\vspace{-1.5em}
%\end{figure}
Figs.~\ref{fig::SR_vs_Gamma}-\ref{fig::SR_vs_Gammae} show the sum-rate of various schemes versus  $\Gamma$ and $\Gamma_e$, respectively. 
The results reveal that the proposed scheme consistently surpasses all baseline approaches in sum-rate across different $\Gamma$ and $\Gamma_e$. This indicates that the proposed scheme can achieve a higher secrecy rate lower bound for fixed $\Gamma$ and/or $\Gamma_e$, thereby significantly boosting system security.
In addition, we plot the analytical lower bound of the secrecy rate as defined in \eqref{eq::PLS_level}, from which the sum-rate and the lower bound on the secrecy rate under specific secrecy requirements for all users can be obtained. For example, with $\Gamma=10$ dB and $\Gamma_e=0$ dB, the proposed scheme achieves a sum-rate of $16.25$ bps/Hz, while the system lower bound for the secrecy rate for all users is $7.38$ bps/Hz. Moreover, in Fig.~\ref{fig::SR_vs_Gamma}, it can be observed that a more stringent communication SINR requirement leads to a higher achievable sum-rate. As $\Gamma$ increases, the performance gap between the ``Comm only'' baseline and the proposed approach narrows, as nearly all the transmit power is utilized to meet the SINR requirements. In Fig.~\ref{fig::SR_vs_Gammae}, it is evident that a less stringent  requirement on secrecy leakage to the target (i.e., $\Gamma_e$ increases), the achievable sum-rate increases initially and then remains relatively unchanged. This is due to the fact that more power can be allocated to information signals with a relaxed secrecy requirement. However, with a limited transmit power budget and radar SNR constraints, the transmit power allocated to the communication information is bounded. In addition, the performance improvement from the ``Random phase'' approach shows only a slight increase with higher $\Gamma$ or $\Gamma_e$, which is attributed to the limited DoFs.

\begin{figure}[t]
	\begin{center}
		\includegraphics[width= 2.7 in]{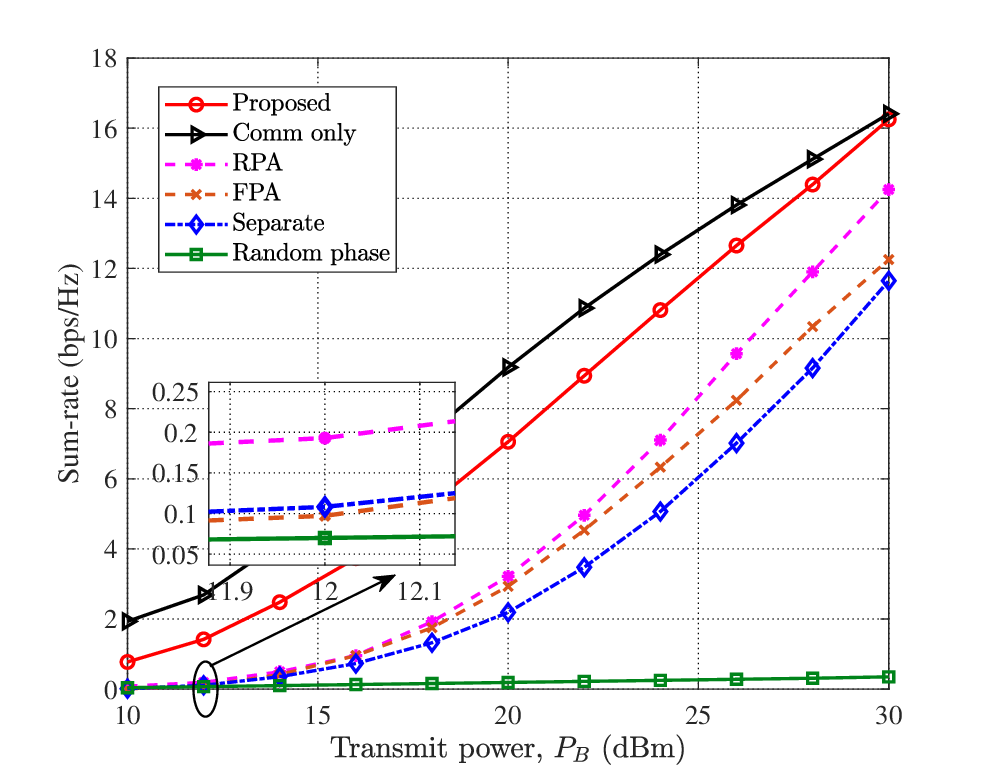}
		\caption{Sum-rate versus transmit power budget $P_B$.} \label{fig::SR_vs_P}
	\end{center}
	\vspace{-1.5em}
\end{figure}
Fig.~\ref{fig::SR_vs_P} compares the achievable sum-rate for various baseline approaches versus $P_B$. The results show that the sum-rate for all schemes increases monotonically with $P_B$. Furthermore, the performance gap between the proposed approach and ``Comm only''  method narrows as $P_B$ increases, since a greater portion of the transmit power is dedicated to improving the performance for communication users while still meeting a fixed radar SNR requirement. Moreover, the sum-rate achieved by the ``Random phase'' approach increases only marginally with increasing $P_B$, since the signals reflected by the RIS propagate in random directions, resulting in low received power levels. In addition, the proposed scheme achieves substantial sum-rate gains compared to the  ``Separate'', ``RPA'', ``FPA'', and ``Random phase'' approaches, illustrating the advantages of jointly designing the transmit/receive beamformers, RIS phase shifts, and positions of MAs. This results in a higher lower bound on the secrecy rate, thereby enhancing the security performance of the considered system.

\begin{figure}[t]
	\begin{center}
		\includegraphics[width= 2.7 in]{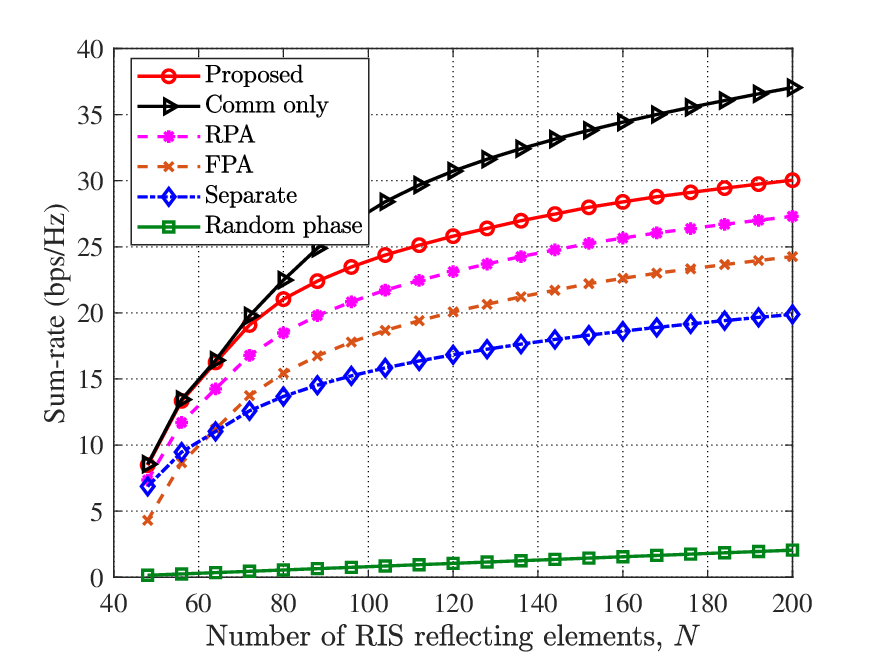}
		\caption{Sum-rate versus number of RIS elements $N$.} \label{fig::SR_vs_N}
	\end{center}
	\vspace{-1.5em}
\end{figure}
Fig.~\ref{fig::SR_vs_N} illustrates the sum-rate comparison for various approaches against the number of RIS reflecting elements $N$. As can be observed, the achievable sum-rate for all baseline schemes increases as $N$ becomes larger due to the increased DoFs available to manipulate the propagation environment. Moreover, the proposed MA-assisted scheme always outperforms the benchmark schemes within the considered range of reflecting elements. For example, when $N$ equals $200$, the maximum sum-rate achieved by the ``FPA'' scheme is $24.26$ bps/Hz, while that achieved by the proposed scheme is approximately $30.04$ bps/Hz, representing an improvement of about $23.8\%$.

Fig.~\ref{fig::SR_vs_A} depicts the sum-rate for different schemes against the normalized size of the receive region $A/\lambda$. The sum-rate increases for the ``Proposed'', ``Comm only'', and ``Separate'' schemes as $A/\lambda$ increases. This is because a larger receive region offers more DoFs for position optimization of the MAs, allowing them to move to locations with improved channel conditions. Nevertheless, it should be noted that the performance enhancement due to the expanded receive region is constrained, and the sum-rate of the MA scheme remains relatively constant as $A/\lambda$ exceeds $3.5$. The results also shows a significant sum-rate improvement for the proposed MAs-aided scheme over other baseline schemes, indicating its ability to achieve a higher secrecy rate lower bound for fixed $\Gamma$ and/or $\Gamma_e$, thus enhancing the security.

\begin{figure}[t]
	\begin{center}
		\includegraphics[width= 2.7 in]{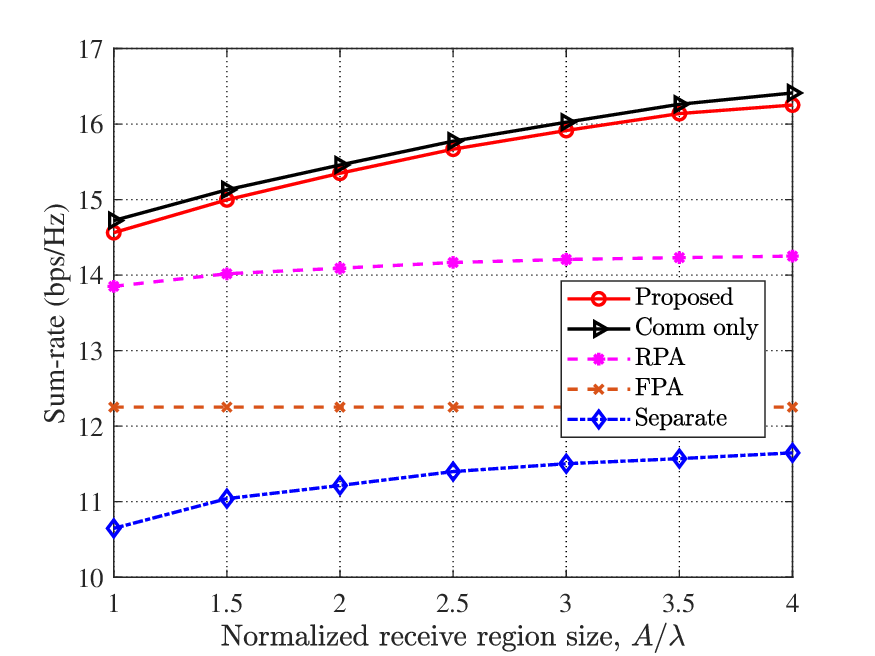}
		\caption{Sum-rate versus normalized receive region size.} \label{fig::SR_vs_A}
	\end{center}
	\vspace{-1.5em}
\end{figure}

\begin{figure}[t]
	\begin{center}
		\includegraphics[width= 2.7 in]{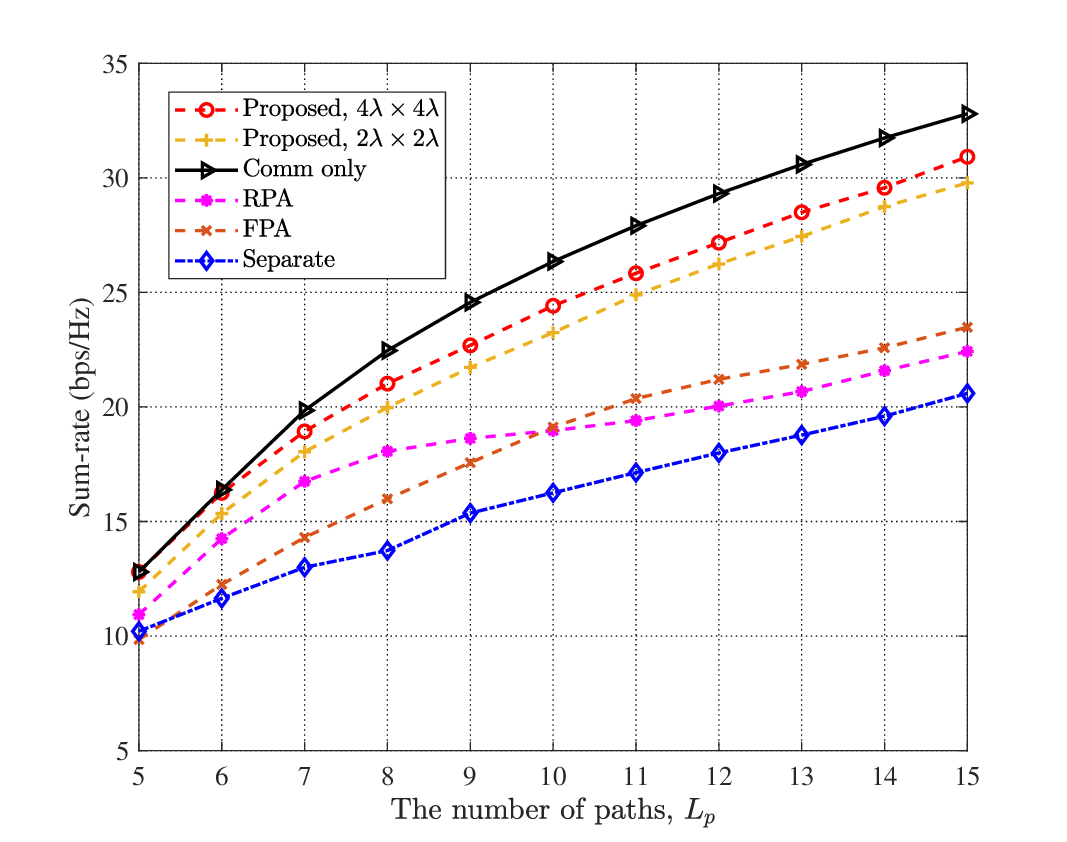}
		\caption{Sum-rate versus number of channel paths.} \label{fig::SR_vs_Lp}
	\end{center}
	\vspace{-1.5em}
\end{figure}
In Fig.~\ref{fig::SR_vs_Lp}, we compare the achievable sum-rate for the MA system with different number of paths, i.e., $L_p$. The results show that the sum-rate increases for all schemes as the number of paths increases. This is because stronger small-scale fading, which occurs with more paths, provides more pronounced channel spatial variation and thus enhances the sum-rate. Furthermore, the proposed MAs-aided approach consistently outperforms all baseline schemes, with the performance gaps widening as $L_p$ increases. Specifically, when $L_p = 5$, the performance gaps for the MA-aided approach with a $4\lambda$ receive region over the RPA and FPA schemes are $17\%$ and $29.95\%$, respectively; these gaps expand to $37.85\%$ and $31.69\%$ when $L_p = 15$.

\section{Conclusion}
In this paper, an MAs-assisted secure transmission scheme for RIS-ISAC systems was investigated, where an eavesdropping target attempts to intercept secrecy data. We formulated an optimization problem aimed at maximizing the sum-rate of all users by jointly optimizing the transmit beamformer, RIS reflection coefficients, receive filter, and MA positions. To tackle the highly non-convex problem, a two-layer penalty-based algorithm was proposed. Specifically, the inner layer alternately solved the penalized optimization problem, while the outer layer updated the penalty factor over iterations. Simulation results confirmed the effectiveness of the proposed MAs-assisted approach in enhancing security performance. Moreover, it was shown that enlarging the size of the receive region, increasing the number of paths, and adding more RIS elements demonstrate a further boost to the sum-rate of MAs-aided systems, thereby enhancing secrecy performance under given secrecy constraints.
%改
%In addition, the simulation results also showed the capability of the RIS in improving sensing and the physical layer security of ISAC systems.

\appendices

\section{Proof of Lemma 1} \label{app0}
Let us consider the inequality as follows
\begin{equation} \small \label{eq::tri}
	\left\|(\boldsymbol{\Lambda}_k - \mathbf{Q}_k)^{\frac{1}{2}}\mathbf{f}_k(\mathbf{u}_k) - (\boldsymbol{\Lambda}_k - \mathbf{Q}_k)^{\frac{1}{2}}\mathbf{f}_k(\mathbf{u}_k^{(t)})\right\|^2 \geq 0,
\end{equation}
where $\boldsymbol{\Lambda}_k \triangleq \lambda_{\max,3,k}\mathbf{I}_{L_k^r}$ with $\lambda_{\max,3,k}$ being the maximum eigenvalue of $\mathbf{Q}_k$. Since the matrix $\boldsymbol{\Lambda}_k - \mathbf{Q}_k$ is positive semidefinite by construction, we expand \eqref{eq::tri} as
\begin{equation} \small
	\begin{aligned}
		&\mathbf{f}_k^{\rm H}(\mathbf{u}_k)(\boldsymbol{\Lambda}_k - \mathbf{Q}_k)\mathbf{f}_k(\mathbf{u}_k) + \mathbf{f}_k^{\rm H}(\mathbf{u}_k^{(t)}) (\boldsymbol{\Lambda}_k - \mathbf{Q}_k) \mathbf{f}_k(\mathbf{u}_k^{(t)}) \\ &\quad \quad - 2\Re\{\mathbf{f}_k^{\rm H}(\mathbf{u}_k) (\boldsymbol{\Lambda}_k - \mathbf{Q}_k) \mathbf{f}_k(\mathbf{u}_k^{(t)})\} \geq 0.
	\end{aligned}
\end{equation}
By isolating the term $\mathbf{f}_k^{\rm H}(\mathbf{u}_k) \mathbf{Q}_k \mathbf{f}_k(\mathbf{u}_k)$, we obtain the bound in \eqref{eq::lemma_2}, this completes the proof.

\section{Derivation of $\nabla\hat{\psi}_k(\mathbf{u}_k)$ and $\nabla^2\hat{\psi}_k(\mathbf{u}_k)$} \label{app1}
% Appendix beginning
Let $\varsigma_{k,i}^{(t)}$ denotes the $i$-th element of $\boldsymbol{\varsigma}_k^{(t)}$, we have
% $\hat{\psi}_k(\mathbf{u}_k)$ can be rewritten as
\begin{equation} \small
	\hat{\psi}_k(\mathbf{u}_k) = 2 \sum_{i=1}^{L_k^r} \left|\varsigma_{k,i}^{(t)}\right|\cos\left(\nu^{(t)}_{k,i}(\mathbf{u}_k)\right) + \varepsilon_5,
\end{equation}
with $\nu^{(t)}_{k,i}(\mathbf{u}_k) \triangleq -\frac{2\pi}{\lambda}\rho^r_{k,i}(\mathbf{u}_k) + \angle\varsigma_{k,i}^{(t)}$ and $\rho^r_{k,i}(\mathbf{u}_k)$ is given by \eqref{eq::receive_FRV}.
The gradient vector $\nabla\hat{\psi}_k(\mathbf{u}_k)$ and Hessian matrix $\nabla^2\hat{\psi}_k(\mathbf{u}_k)$ can be represented as $\nabla\hat{\psi}_k(\mathbf{u}_k) = [\frac{\partial \hat{\psi}_k(\mathbf{u}_k)}{\partial x_k} \quad \frac{\partial \hat{\psi}_k(\mathbf{u}_k)}{\partial y_k}]^{\rm T}$ and $\nabla^2\hat{\psi}_k(\mathbf{u}_k) = \bigg[\begin{array}{cc}
	\frac{\partial^2 \hat{\psi}_k(\mathbf{u}_k)}{\partial x_k\partial x_k} & \frac{\partial^2 \hat{\psi}_k(\mathbf{u}_k)}{\partial x_k\partial y_k} \\
	\frac{\partial^2 \hat{\psi}_k(\mathbf{u}_k)}{\partial y_k\partial x_k} & \frac{\partial^2 \hat{\psi}_k(\mathbf{u}_k)}{\partial y_k\partial y_k} \end{array}\bigg]$.
Thus, we have

\begin{subequations} \label{eq::first_order} \small
	\begin{align}
		&\frac{\partial \hat{\psi}_k(\mathbf{u}_k)}{\partial x_k} =  \frac{4\pi}{\lambda} \sum_{i=1}^{L_k^r} \left|\varsigma_{k,i}^{(t)}\right|\sin\left(\nu^{(t)}_{k,i}(\mathbf{u}_k)\right)\sin\theta^r_{k,i}\cos\phi^r_{k,i}, \\
		&\frac{\partial \hat{\psi}_k(\mathbf{u}_k)}{\partial y_k} = 
		\frac{4\pi}{\lambda} \sum_{i=1}^{L_k^r} \left|\varsigma_{k,i}^{(t)}\right|\sin\left(\nu^{(t)}_{k,i}(\mathbf{u}_k)\right)\cos\theta^r_{k,i},
	\end{align}
\end{subequations}

\vspace{-1 em}
\begin{subequations} \label{eq::second_order} \small
	\begin{align}
		&\!\!\!\frac{\partial^2 \hat{\psi}_k(\mathbf{u}_k)}{\partial x_k\partial x_k} \!=\! -\frac{8\pi^2}{\lambda^2}\!\! \sum_{i=1}^{L_k^r}\! \left|\varsigma_{k,i}^{(t)}\right|\!\cos\left(\!\nu^{(t)}_{k,i}(\mathbf{u}_k)\!\right)\sin^2\!\theta^r_{k,i}\cos^2\!\phi^r_{k,i}, \\
		&\frac{\partial^2 \hat{\psi}_k(\mathbf{u}_k)}{\partial x_k\partial y_k} = \frac{\partial^2 \hat{\psi}_k(\mathbf{u}_k)}{\partial y_k\partial x_k} = \nonumber \\
		&-\!\frac{8\pi^2}{\lambda^2}\! \sum_{i=1}^{L_k^r} \left|\varsigma_{k,i}^{(t)}\right|\!\cos\left(\nu^{(t)}_{k,i}(\mathbf{u}_k)\right)\!\sin\theta^r_{k,i}\!\cos\phi^r_{k,i}\!\cos\theta^r_{k,i}, \\
		& 	\frac{\partial^2 \hat{\psi}_k(\mathbf{u}_k)}{\partial y_k\partial y_k} \!=\! -\frac{8\pi^2}{\lambda^2} \sum_{i=1}^{L_k^r} \left|\varsigma_{k,i}^{(t)}\right|\cos\left(\nu^{(t)}_{k,i}(\mathbf{u}_k)\right)\cos^2\theta^r_{k,i}.
	\end{align}
\end{subequations}

\section{Construction of $\delta_k$} \label{app2}
Since we have
\begin{equation} \small
	\begin{aligned}
		\left\|\nabla^2\hat{\psi}_k(\mathbf{u}_k)\right\|_2^2 & \!\leq\! 	\left\|\nabla^2\hat{\psi}_k(\mathbf{u}_k)\right\|_F^2  \!\leq\! 4\Big(\frac{8\pi^2}{\lambda^2}\sum_{i=1}^{L_k^r}\left|\varsigma_{k,i}^{(t)}\right|\Big)^2,
	\end{aligned}
\end{equation}
and
\begin{equation} \small
	\left\|\nabla^2\hat{\psi}_k(\mathbf{u}_k)\right\|_2\mathbf{I}_2 \succeq \nabla^2\hat{\psi}_k(\mathbf{u}_k),
\end{equation}
thus we can select $\delta_k$ as
\begin{equation} \label{eq::delta} \small
	\delta_k = \frac{16\pi^2}{\lambda^2}\sum_{i=1}^{L_k^r}\left|\varsigma_{k,i}^{(t)}\right|,
\end{equation}
which is satisfied the following condition
\begin{equation} \small
	\delta_k\mathbf{I}_2 \succeq \left\|\nabla^2\hat{\psi}_k(\mathbf{u}_k)\right\|_2\mathbf{I}_2 \succeq \nabla^2\hat{\psi}_k(\mathbf{u}_k).
\end{equation}

\bibliographystyle{IEEEtran}
\bibliography{./bibtex/IEEEabrv,./bibtex/IEEEexample}

% that's all folks%
\end{document}